\pdfoutput=1
\documentclass[10pt]{article}

\usepackage[section]{placeins}

\usepackage{bibspacing}
\usepackage[paper=letterpaper,centering]{geometry}
\usepackage{amsmath,amscd}
\usepackage{amsthm}

\usepackage{graphics}
\usepackage{lhelp}
\usepackage{amsmath,amscd}
\usepackage{pstricks,pst-node,pst-grad,pst-plot}
\usepackage{amsfonts}
\usepackage{epsfig, color}
\usepackage{latexsym}
\usepackage{url}

\usepackage[linesnumbered,ruled,vlined]{algorithm2e}
\usepackage{tikz}

\def\A {\ensuremath{\mathbb{A}}}
\def\B {\ensuremath{\mathbb{B}}}
\def\C {\ensuremath{\mathbb{C}}}

\def\K {\ensuremath{\mathbf{k}}}

\def\N {\ensuremath{\mathbb{N}}}
\def\Q {\ensuremath{\mathbb{Q}}}
\def\R {\ensuremath{\mathbb{R}}}

\def\Z {\ensuremath{\mathbb{Z}}}

\newtheorem{Notation}{Notation}
\newtheorem{Theorem}{Theorem}
\newtheorem{Proposition}{Proposition}
\newtheorem{Definition}{Definition}
\newtheorem{Lemma}{Lemma}

\newtheorem{Remark}{Remark}

\newcommand{\lc}[1]{\mbox{{\rm lc}$(#1)$}}

\newcommand{\ideal}[1]{\langle#1\rangle}
\newcommand{\init}[1]{\mbox{{\rm init}$(#1)$}}

\newcommand{\mdeg}[1]{\mbox{{\rm mdeg}$(#1)$}}
\newcommand{\mvar}[1]{\mbox{{\rm mvar}$(#1)$}}

\newcommand{\res}[1]{\mbox{{\rm res}$(#1)$}}

\newcommand{\sat}[1]{\mbox{{\rm sat}$(#1)$}}

\newcommand{\coeff}[1]{\mbox{{\rm coeff}$(#1)$}}
\renewcommand{\min}[1]{\mbox{{\rm min}$(#1)$}}
\renewcommand{\max}[1]{\mbox{{\rm max}$(#1)$}}

\newcommand{\Maple}{{\sc  Maple}}

\newcommand{\RegularChains}{{\tt  Regu\-lar\-Chains}}

\renewcommand{\gcd}[1]{\mbox{{\rm gcd}$(#1)$}}

\newcommand{\discrim}[1]{\mbox{{\rm discrim}$(#1)$}}

\newcommand{\limit}[1]{\mbox{${\lim}(#1)$}}

\newcommand{\carr}[1]{\mbox{{\rm carr}$(#1)$}}
\newcommand{\ord}[1]{\mbox{{\rm ord}$(#1)$}}

\newcommand{\primpart}[1]{\mbox{{\rm primpart}$(#1)$}}

\newcommand{\NewtonPuiseux}[1]{\mbox{{\sf NewtonPuiseux}$(#1)$}}
\newcommand{\eval}[1]{\mbox{{\rm eval}$(#1)$}}
\newcommand{\LimitPointsAtZero}[1]{\mbox{{\sf LimitPointsAtZero}$(#1)$}}

\newcommand{\NewtonPolygon}[1]{\mbox{{\sf NewtonPolygon}$(#1)$}}

\newcommand{\NewPolynomial}[1]{\mbox{{\sf NewPolynomial}$(#1)$}}
\newcommand{\SegmentPoly}[1]{\mbox{{\sf SegmentPoly}$(#1)$}}

\newcommand{\NonzeroTerm}[1]{\mbox{{\sf NonzeroTerm}$(#1)$}}
\newcommand{\ConstructParametrization}[1]{\mbox{{\sf ConstructParametrization}$(#1)$}}

\providecommand{\abs}[1]{\lvert#1\rvert}

\newcommand{\citep}[1]{\mbox{ \cite{#1} }}
\newif\ifcomment
\commentfalse

\newif\iflongversion
\longversiontrue

\newif\ifready
\readyfalse

\begin{document}

\begin{center}
  {\Large\bf 
    An Algorithm for Computing the Limit Points of the Quasi-component of a Regular Chain
  }
\mbox{}\\[11pt]
{\large
  Parisa Alvandi, Changbo Chen, Marc Moreno Maza
}
\mbox{}\\[5pt] 
ORCCA, University of Western Ontario (UWO) \\
London, Ontario, Canada \\
{\tt palvandi@uwo.ca, changbo.chen@gmail.com, moreno@csd.uwo.ca}\\[5pt]
\end{center}

\begin{abstract}
For a regular chain $R$, we propose an algorithm
which computes the (non-trivial) limit points of the quasi-component
of $R$, that is, the set $\overline{W(R)} \setminus W(R)$.
Our procedure relies on Puiseux series expansions
and does not require to compute a system
of generators of the saturated ideal of $R$.
We focus on the case where this saturated ideal has dimension one
and we discuss extensions of this work in higher dimensions.
We provide experimental results illustrating the benefits 
of our algorithms.
\end{abstract}

\section{Introduction}

The theory of regular chains, since its introduction
by J.F. Ritt~\cite{Ritt32}, 
has been applied successfully in many areas including
\iflongversion
parametric algebraic systems~\cite{ChGa91a},
\fi
differential systems~\cite{ChGa93,BLOP95,Hubert00},
difference systems~\cite{GaoHoevenLuoYuan2009},
intersection multiplicity~\cite{DBLP:conf/casc/MarcusMV12},
unmixed decompositions~\cite{Kalk98} and 
primary decomposition~\cite{ShiYo96}
of polynomial ideals,
cylindrical algebraic decomposition~\cite{CMXY09},
parametric~\cite{yhx01}
and non-parametric~\cite{DBLP:journals/jsc/ChenDMMXX13} semi-algebraic systems.
Today, regular chains are at the core of 
algorithms for triangular decomposition of polynomial systems,
which are available in several software   
packages~\cite{LeMoXi04,DingkangWang,EpsilonWang}.
Moreover, these algorithms provide back-engines for 
computer algebra system front-end solvers,
such as {\sc Maple}'s {\tt solve} command.

One of the algorithmic strengths of the theory of regular chains
is its {\em regularity test} procedure. Given a 
polynomial $p$ and a regular chain $R$, both 
in a multivariate polynomial ring ${\K}[X_1, \ldots, X_n]$
over a field {\K}, this procedure computes regular chains 
$R_1, \ldots, R_e$ such that $R_1, \ldots, R_e$
is a decomposition of $R$ in some technical
sense~\footnote{The radical of the saturated
ideal of $R$ is equal to the intersection of the radicals of the saturated
ideals of $R_1, \ldots, R_e$.} and for each $1 \leq i \leq e$
the polynomial $p$ is either null or regular modulo
the saturated ideal of $R_i$.
Thanks to the D5 Principle~\cite{D5}, this regularity test
avoids factorization into irreducible polynomials
and involves only polynomial GCD and resultant computations.

One of the technical difficulties of this theory, however, 
is the fact that regular chains do not fit well in the 
``usual algebraic-geometric dictionary'' (Chapter 4, \cite{DCJLDOS97}).
Indeed, the ``good'' zero set encoded by a regular chain $R$
is a constructible set $W(R)$, called the {\em quasi-component}
of $R$, which does not correspond exactly to 
the ``good'' ideal  encoded by $R$, namely \sat{R}, the {\em saturated
ideal} of $R$. In fact, the affine variety defined by \sat{R}
equals $\overline{W(R)}$, that is, the Zariski closure of $W(R)$.

For this reason, a decomposition algorithm, such as the one
of M.~Kalkbrener~\cite{Kalk98} (which, for an input polynomial
ideal ${\cal I}$ computes regular chains 
$R_1, \ldots, R_e$ such that $\sqrt{{\cal I}}$ equals
the intersection of the radicals of the saturated
ideals of $R_1, \ldots, R_e$) can not be seen as a decomposition
algorithm for the variety $V({\cal I})$.
Indeed, the output of Kalkbrener's algorithm yields 
$V({\cal I}) = \overline{W(R_1)} \, \cup \, \cdots \, \cup \,  \overline{W(R_e)}$
while a decomposition of the form 
$V({\cal I}) = {W(R_1)} \, \cup \, \cdots \, \cup \,  {W(R_f)}$
would be more explicit.

Kalkbrener's decompositions, and in fact all decompositions
of differential ideals~\cite{ChGa93,BLOP95,Hubert00} 
raise another notorious issue: the {\em Ritt problem},
stated as follows.
Given two regular chains (algebraic or differential) $R$ and $S$,
check whether the inclusion of saturated ideals $\sat{R} \subseteq \sat{S}$
holds or not.
In the algebraic case, this inclusion can be tested by computing
a set of generators of $\sat{R}$ , using Gr\"obner bases.
In practice, this solution is too expensive
for the purpose of removing redundant components in Kalkbrener's decompositions
and only some criteria are applied~\cite{DBLP:journals/jsc/LemaireMPX11}.
In the differential case , there has not even an algorithmic solution.

In the algebraic case, both issues would be resolved if one
would have a practically efficient procedure with the
following specification: for the regular chain $R$
compute regular chains $R_1, \ldots, R_e$ such that
we have $\overline{W(R)} = W(R_1) \, \cup \, \cdots \, \cup \, W(R_e)$.
If in addition, such procedure does not require a system of generators of $\sat{R}$,
this might suggest a solution in the differential
version of the Ritt problem.

In this paper, we propose a solution to this
algorithmic quest, in the algebraic case. 
To be precise, our procedure computes the 
{\em non-trivial limit points} of the quasi-component $W(R)$, that is,
the set $\limit{W(R)} := \overline{W(R)} \setminus W(R)$.
\iflongversion
This turns out to be
 $\overline{W(R)} \, \cap \, V(h_R)$, where $V(h_R)$
is the hypersurface defined by the product of the initials
of $R$.
\fi
We focus on the case where the saturated ideal of $R$  has dimension one.
In Section~\ref{sec:discussion}, we sketch a solution in higher dimension.

When the regular chain $R$ consists of a single polynomial $r$,
primitive w.r.t. its main variable,
one can easily check that $\limit{W(R)} = V(r, h_R)$ holds.
Unfortunately, there is no generalization of this result
when $R$ consists of several polynomials, unless $R$ enjoys
remarkable properties, such as being a 
{\em primitive regular chain}~\cite{DBLP:journals/jsc/LemaireMPX11}.
To overcome this difficulty, it becomes necessary to view $R$
as a ``parametric representation'' of the quasi-component $W(R)$.
In this setting, the points of $\limit{W(R)}$  can be computed
as limits (in the usual sense of the Euclidean 
topology~\footnote{This identification of the closures of $W(R)$ in Zariski
topology and the Euclidean topology holds when {\K} is {\C}.})
of sequences of points along  ``branches'' (in the sense
of the theory of algebraic curves) of $W(R)$ .
It turns out that these limits can be obtained as constant terms 
of convergent Puiseux series defining the ``branches''  of $W(R)$
in the neighborhood of the points of interest.

Here comes the main technical difficulty of this approach.
When computing a particular point of $\limit{W(R)}$,
one needs to follow one branch per defining equation of $R$.
Following a branch means computing a truncated 
Puiseux expansion about a point.
Since the equation of $R$ defining a given variable, say $X_j$,
depends on the equations of $R$ defining the variables $X_{j-1}, X_{j-2}, \ldots$, 
the truncated Puiseux expansion for  $X_j$
is defined by an equation whose coefficients
involve the truncated Puiseux expansions for $X_{j-1}, X_{j-2}, \ldots$.

From Sections~\ref{sec:theory}  to~\ref{sec:principle}, we show
that this principle indeed computes the desired limit points.
In particular, we introduce the notion of a 
{\em system of Puiseux parametrizations of a regular chain}, see Section~\ref{sec:theory}.
This allows to state in Theorem~\ref{thrm:PuiseuxParamRC} 
a concise formula for $\limit{W(R)}$ in terms of this latter notion.
Then, we estimate to which accuracy 
one needs to effectively compute such a system of Puiseux parametrizations
in order to deduce $\limit{W(R)}$, see Theorem~\ref{Theorem:bound-general}
in Section~\ref{sec:bounds}.

In Section~\ref{section}, we report on a preliminary implementation
of the algorithms presented in this paper.
We evaluate our code by applying it to the question of
removing redundant components in Kalkbrener's decompositions
and observe the benefits of this strategy.

In order to facilitate the presentation of those technical
materials, we dedicate Section~\ref{sec:limitpoints} to the case of
regular chains in $3$ variables.
Section~\ref{sec:preliminaries} briefly reviews
notions from the theories of regular chains and
algebraic curves.
We conclude this introduction with a detailed example.

Consider the regular chain 
 $R=\{r_1,r_2\}\subset \K[X_1,X_2,X_3]$ 
with  $r_1=X_1X_2^2+X_2+1,r_2=(X_1+2)X_1X_3^2+(X_2+1)(X_3+1)$. 
Then, we have $h_R=X_1(X_1+2)$. 
To determine \limit{R}, we need to compute Puiseux series
expansions of $r_1$ about $X_1=0$ and $X_1=-2$.
We start with $X_1=0$. The two Puiseux expansions of $r_1$ about $X_1=0$ 
are:
$$\begin{array}{c}[X_1=T,X_2={\frac {-{T}^{2}-T}{T}}+O(T^2)],$$\\ 
$$ [X_1=T,X_2={\frac{-1+{T}^{2}+T}{T}}+O(T^2)].\end{array}$$
The second expansion does not result in a new limit point. 
After, substituting the first expansion into $r_2$, we have:
$$\begin{array}{ll}\hspace*{-0.2cm}r_2'&=  r_2(X_1=T,X_2={\frac {-{T}^{2}-T}{T}}+O(T^2),X_3)\\&= T\left( \left( T+2 \right) {X_3}^{2}+ \left( O(T^3) {-5\,{T}^{2}-2\,{T}-1}+1 \right)  \left( X_3+1 \right) \right). \end{array}$$ 
Now, we compute Puiseux series
expansions of $r_2'$ which are
$$\begin{array}{c}[T=T,X_3=1+T+O(T^2)],$$\\$$[T=T,  X_3=-1/2-1/4\,T+O(T^2)].\end{array}$$
So the regular chains $\{X_1,X_2+1,X_3-1\}$ and $\{X_1,X_2+1,X_3+1/2\}$
give the limit points of $W(R)$ about $X_1=0$.

Next, we consider $X_1=-2$.
We compute Puiseux series expansions of $r_1$ about the point $X_1=-2$. We have:
$$\begin{array}{c}[X_1=T-2,X_2=1+1/3\,T+O(T^2)],$$\\$$  [X_1=T-2,X_2=-1/2-1/12\,T+O(T^2)].\end{array}$$ 
After substitution into $r_2$, we obtain:
\begin{equation*}
\begin{array}{c}
\begin{array}{ll}r_{12}' &= r_2(X_1=T-2,X_2=1+1/3\,T+O(T^2),X_3) \\&= \left( T-2 \right) T{X_3}^{2}+ \left( 2+1/3\,T+O(T^2) \right)  \left( X_3+1 \right)\end{array} \\
\begin{array}{ll}r_{22}'&= r_2([X_1=T-2,X_2=-1/2-1/ 12\,T+O(T^2)) \\ &= \left( T-2 \right) T{X_3}^{2}+ \left( 1/2-1/12\,T+O(T^2) \right)  \left( X_3+1 \right).\end{array}
\end{array}
\end{equation*}
So those Puiseux expansions of $r_{12}'$ and $r_{22}'$ about $T=0$ which result in a limit point are as follows:
\begin{itemize}
\item[i)]  for $r_{12}'$: $[T=T,X_3=  \frac{{T}^{2}-T}{{T}}+O(T^2)]$
\item[ii)]for $r_{22}'$:  $ [T=T,X_3= \frac{4
\,{T}^{2}-T }{T}+O(T^2)]$
\end{itemize}
Thus, the limit points of $R$ about the point $X_1=-2$ can be 
represented by the regular chains $\{ X_1+2,X_2-1,X_3+1\}$ and $\{X_1+2,X_2+1/2,X_3+1\}$.

One can check that a triangular decomposition of the 
system $R \,  \cup \, \{X_1\}$
is $\{X_2+1,X_1\}$ and, thus, does not yield $\limit{W(R)} \, \cap \, V(X_1)$,
but in fact a superset of it.

\section{Preliminaries}
\label{sec:preliminaries}

This section is a brief review of various notions
from the theories of regular chains,  algebraic curves
and topology.
For these latter subjects, our references are
the textbooks of R.J. Walker~\cite{RW78},
G. Fischer~\cite{GF01} and J. R. Munkres~\cite{Mun00}.
The notations and hypotheses introduced in this section
are used throughout the sequel of the paper.

\smallskip\noindent{\small \bf Multivariate polynomials.}
Let ${\K}$ be a field which is algebraically closed.
Let $X_1 < \cdots < X_s$ be $s \geq 1$ ordered variables.
We denote by ${\K}[X_1, \ldots, X_s]$ the ring of 
polynomials in the  variables $X_1, \ldots, X_s$ 
and with coefficients in ${\K}$.
For a non-constant polynomial $p\in {\K}[X_1, \ldots, X_s]$,
the grea\-test variable in $p$ 
is called {\em main variable} of $p$, denoted by $\mvar{p}$,
and the leading coefficient of $p$ w.r.t. $\mvar{p}$
is called {\em initial} of $p$, denoted by $\init{p}$.

\smallskip\noindent{\small \bf Zariski topology.}
We denote by ${\A}^s$ the {\em affine $s$-space} over ${\K}$.
An {\em affine variety} of  ${\A}^s$ 
is the set of common zeroes of a collection 
$F \subseteq {\K}[X_1, \ldots, X_s]$ of polynomials.
The {\em Zariski topology} on ${\A}^s$  is the topology 
whose closed sets are the affine varieties of ${\A}^s$.
The {\em Zariski closure} of a subset $W \subseteq {\A}^s$
is the intersection of all affine varieties containing $W$.
This is also the set of common zeroes of the polynomials
in ${\K}[X_1, \ldots, X_s]$ vanishing at any point of $W$.

\smallskip\noindent{\small \bf Relation between Zariski topology
and the Euclidean topology.}
When ${\K} = {\C}$, the affine space ${\A}^s$ is endowed
with both Zariski topology and the Euclidean topology.
The basic open sets of the Euclidean topology are the 
balls while the basic open sets of Zariski topology
are the complements of hypersurfaces.
\iflongversion
A Zariski closed (resp. open) set is closed (resp. open) 
in the Euclidean topology on ${\A}^s$.
The following properties emphasize the fact that
Zariski topology is coarser than the Euclidean topology:
every nonempty Euclidean open set is Zariski dense and
every nonempty Zariski open set is dense in the 
Euclidean topology on ${\A}^s$.
However, the closures of a constructible set in Zariski topology
and the Euclidean topology are equal.
More formally, we have the following (Corollary 1 in I.10 of~\cite{MUM99}) key result.
\fi
Let $V \subseteq {\A}^s$ be an irreducible affine variety
and $U \subseteq V$ be open in the Zariski topology
induced on $V$.
Then the closure of $U$ in Zariski topology
and the closure of $U$ in the Euclidean topology
are both equal to $V$.

\smallskip\noindent{\small \bf Limit points.}
Let $(X, {\tau})$ be a topological space.
A point $p \in X$ is a {\em limit} of a 
sequence $(x_n, n\in {\N})$ of points of $X$ if, 
for every neighborhood $U$ of $p$, 
there exists an $N$ such that, for every $n \geq N$, 
we have $x_n \in U$; when this holds we write 
$\lim_{n \rightarrow  \infty} \, x_n = p$.
If $X$ is a Hausdorff space 
then limits of sequences are unique, when they exist.
Let $S \subseteq X$ be a subset.
A point $p \in X$ is a {\em limit point} of $S$ if every neighborhood 
of $p$ contains at least one point of $S$ different from $p$ itself.
\iflongversion
Equivalently, $p$ is a limit point of $S$ if it 
is in the closure of $S \setminus \{p\}$.
In addition, the closure of $S$ is equal to the union
of $S$ and the set of its limit points.
\fi
If the space $X$ is sequential, 
and in particular if $X$ is a metric space,
the point $p$ is a  limit point of $S$ if and only 
if there exists a  sequence $(x_n, n\in {\N})$ of points of 
$S \setminus \{ p \}$ with $p$ as limit.
In practice, the ``interesting'' limit points of $S$
are those which do not belong to $S$.
For this reason, we call
such limit points {\em non-trivial} and we denote
by $\limit{S}$ the set of non-trivial limit points of $S$.

\smallskip\noindent{\small \bf Regular chain.}
A set $R$ of non-constant polynomials in ${\K}[X_1, \ldots, X_s]$ is called 
a {\em triangular set},
if for all $p,q \in R$ with $p \neq q$ we have $\mvar{p} \neq \mvar{q}.$
For a nonempty triangular set $R$, we define the 
{\em saturated ideal} \sat{R} of $R$
to be the ideal $\ideal{R}:h_R^{\infty}$, 
where $h_R$ is the product of the initials of the polynomials in $R$. 
The empty set is also regarded as a triangular set, whose saturated ideal
is the trivial ideal $\langle 0 \rangle$.
From now on, $R$ denotes a triangular set of ${\K}[X_1, \ldots, X_s]$.
The ideal \sat{R} has several properties, in particular
it is  unmixed~\cite{BLM01c}.
We denote its height by $e$, thus  \sat{R} has dimension $s - e$. 
Without loss of generality, we assume  that 
${\K}[X_1, \ldots, X_{s-e}] \cap \sat{R}$ is the trivial
ideal $\langle 0 \rangle$.
For all  $1 \leq i \leq e$, 
we denote by $r_i$ the polynomial of $R$ whose
main variable is $X_{i+s-e}$ and by $h_i$ the initial of $r_i$.
Thus $h_R$ is the product  $h_1 \cdots h_{e}$.
We say that $R$ is a {\em regular chain}
whenever $R$ is empty or $\{ r_1, \ldots, r_{e-1} \}$
is a regular chain and $h_e$ is regular modulo
the saturated ideal \sat{\{ r_1, \ldots, r_{e-1} \}}.
The regular chain $R$ is said {\em strongly normalized}
whenever  $h_R \in {\K}[X_1, \ldots, X_{s-e}]$ holds.
If $R$ is not strongly normalized, one can compute
a regular chain $N$ which is strongly normalized
and such that $\sat{R} = \sat{N}$ and
$V(h_N) = V(\widehat{h_R})$ both hold, 
where $\widehat{h_R}$ is the iterated resultant
of $h_R$ w.r.t $R$. See~\cite{DBLP:journals/jsc/ChenM12}.

\smallskip\noindent{\small \bf Limit points of the 
quasi-component of a regular chain.}
We denote by $W(R) := V(R) \setminus V(h_R)$ 
the {\em quasi-component} of $R$,
that is, the common zeros of $R$ that do not
cancel $h_R$.
The above discussion implies that
the closure of $W(R)$ in Zariski topology
and the closure of $W(R)$ in the Euclidean topology
are both equal to $V(\sat{R})$, that is, the
affine variety of $\sat{R}$.
We denote by $\overline{W(R)}$ this common closure.
We call {\em limit points} of $W(R)$
the elements of $\limit{W(R)}$.

\smallskip\noindent{\small \bf Rings of formal power series.}
Recall that ${\K}$ is an algebraically closed field.
From now on, we further assume that ${\K}$ is  topologically complete. 
\iflongversion
Hence ${\K}$ may be the field ${\C}$ of complex
numbers but not the algebraic closure of the field ${\Q}$ of rational numbers.
\fi
We denote by ${\K}[[X_1, \ldots, X_s]]$ and
${\K} \langle X_1, \ldots, X_s \rangle$ the rings of formal
and convergent power series in $X_1, \ldots, X_s$ with
coefficients in ${\K}$.
\iflongversion
Note that the ring ${\K} \langle X_1, \ldots, X_s \rangle$
is a subring of ${\K}[[X_1, \ldots, X_s]]$.
\fi
When $s=1$, we write $T$ instead of $X_1$.
\iflongversion
Thus ${\K}[[T]]$ and ${\K} \langle T \rangle$
are the rings of formal
and convergent univariate power series in $T$ and 
coefficients in ${\K}$.
\fi
For $f \in {\K}[[X_1, \ldots, X_s]]$, its {\em order} is defined by
\iflongversion
\begin{equation*}
{\rm ord}(f) = \left\{ \begin{array}{lr}
                        {\rm min} \{  d \mid f_{(d)} \neq 0 \} & {\rm if} \ f \neq 0, \\
                                      \infty                   & {\rm if} \ f = 0.
                        \end{array}
                 \right.
\end{equation*}
where $f_{(d)}$ is the {\em homogeneous part} of $f$ in degree $d$.
\fi 
\iflongversion
Recall that ${\K}[[X_1, \ldots, X_s]]$ is topologically
complete for Krull Topology and 
that ${\K} \langle X_1, \ldots, X_s \rangle$ is a Banach Algebra
for the norm defined by
${ \parallel f \parallel }_{\rho} = {\Sigma}_e \, |a_e| {\rho}^e$
where $f = {\Sigma}_e \, a_e X^e \in {\K}[[ X_1, \ldots, X_s ]]$
and ${\rho} = ({\rho}_1, \ldots, {\rho}_s)  \in {\R}_{> 0}^s$.
\fi
We denote by ${\cal M}_s $
the only maximal ideal of ${\K}[[X_1, \ldots, X_s]]$,
that is, 
\iflongversion
\begin{equation*}
{\cal M}_s= \{ f \in {\K}[[X_1, \ldots, X_s]] \mid {\rm ord}(f)  \geq 1 \}.
\end{equation*}
\fi
Let $f \in {\K}[[ X_1, \ldots, X_s ]]$ with $f \neq 0$.
Let $k \in {\N}$.
We say that $f$ is  
(1) {\em general} in $X_s$ if $f \neq 0 \mod{ {\cal M}_{s-1}}$,
(2) {\em general} in $X_s$ of order $k$ if we have 
      ${\rm ord}(  f  \mod{ {\cal M}_{s-1} }) = k$.

\smallskip\noindent{\small \bf Formal Puiseux series.}
We denote by
${\K} [[ T^{*} ]] \ = \  
                 \bigcup_{n=1}^{\infty} \, {\K} [[ T^{\frac{1}{n}} ]]$
the ring of {\em formal Puiseux series}.
For a fixed ${\varphi} \in {\K} [[ T^{*} ]]$, 
there is an $n \in {\N}_{>0}$
such that ${\varphi} \in {\K} [[ T^{\frac{1}{n}} ]]$. Hence
$
{\varphi}  = \sum_{m=0}^{\infty} \, a_m T^{\frac{m}{n}}$, 
{\rm where} $\ a_m \in {\K}$.
We call {\em order of} ${\varphi}$ 
the rational number defined by
$
{\rm ord}( {\varphi} ) \ = \ {\rm min} \{ \frac{m}{n} \ \mid \  a_m \neq  0 \}
\geq 0.
$
We denote by ${\K} (( T^{*} ))$ 
the quotient field of ${\K} [[ T^{*} ]]$.

\smallskip\noindent{\small \bf Convergent Puiseux series.}
Let ${\varphi} \in {\C} [[ T^{*} ]]$ and $n \in {\N}$ 
such that ${\varphi}=f(T^{\frac{1}{n}})$ with $f\in{\C}[[T]]$ holds.
We say that the Puiseux series ${\varphi}$ is {\em convergent}
if we have ${f} \in {\C} \langle T \rangle$.
Convergent  Puiseux series form an integral
domain denoted by ${\C} \langle T^{*} \rangle$;
its quotient field is denoted 
by ${\C} ( \langle T^{*} \rangle )$.
For every ${\varphi} \in {\C}(( T^{*} ))$, there exist
$n \in {\Z}$, $r \in {\N}_{>0}$ and 
a sequence of complex numbers $a_n, a_{n+1}, a_{n+2}, \ldots$ such that
we have 
\iflongversion
\begin{equation*}
{\varphi} \ = \ \sum_{m=n}^{\infty} \, a_m T^{\frac{m}{r}} \ \ {\rm and}  \ \
a_n \neq 0.
\end{equation*}
\fi 
Then, we define ${\rm ord}({\varphi}) = \frac{n}{r}.$

\smallskip\noindent{\small \bf Puiseux Theorem.}
If {\K} has characteristic zero, 
the field ${\K} (( T^{*} ))$ is the algebraic closure of the field of 
formal Laurent series over  ${\K}$.
Moreover, if ${\K} = {\C}$, 
the field ${\C} ( \langle T^{*} \rangle )$ 
is algebraically closed as well.
From now on, we assume ${\K} = {\C}$.

\smallskip\noindent{\small \bf Puiseux expansion.}
Let $\B=\C((X^{*}))$ or ${\C} ( \langle X^{*} \rangle )$.
Let $f\in\B[Y]$, where $d := \deg(f, Y) >0$.
Let $h := \lc{f, Y}$.
According to Puiseux Theorem, there exists ${\varphi_i}\in \B$, 
$i=1,\ldots,d$,
such that $\frac{f}{h}=(Y-\varphi_1)\cdots(Y-\varphi_d)$.
We call $\varphi_1, \ldots, \varphi_d$
the {\em Puiseux expansions} of $f$ at the origin.

\smallskip\noindent{\small \bf Puiseux parametrization.}
Let $f\in\C\langle X \rangle[Y]$.
A parametrization of $f$ is a 
pair $(\psi(T), \varphi(T))$
of elements of $\C\langle T \rangle$ for some new variable $T$, 
such that (1) $f(\psi(T), \varphi(T))=0$ holds in $\C\langle T \rangle$,
(2) we have $0<\ord{\psi(T)}$, and
(3) $\psi(T)$ and 
$\varphi(T)$ are not both in $\C$.
The parametrization $(\psi(T), \varphi(T))$
is {\em irreducible} if there is no integer $k>1$ such that
both $\psi(T)$ and $\varphi(T)$ are in $\C\langle T^k\rangle$.
We call an irreducible parametrization 
$(\psi(T), \varphi(T))$ of $f$ a {\em Puiseux parametrization}
of $f$, if there exists a positive integer $\varsigma$
such that $\psi(T)=T^{\varsigma}$. 
The index $\varsigma$ is called the {\em ramification index} of 
the parametrization $(T^{\varsigma}, \varphi(T))$.
It is intrinsic to $f$ and $\varsigma\leq \deg(f, Y)$.
Let $z_1,\ldots,z_{\varsigma}$ denote the primitive
roots of unity of order $\varsigma$ in $\C$.
Then $\varphi(z_iX^{1/\varsigma})$, for $i=1,\ldots,\varsigma$, are $\varsigma$
Puiseux expansions of $f$.


\medskip

We conclude this section by a few lemmas which
are immediate consequences of the above review.

\begin{Lemma}
\label{Lemma:limitWR1}
We have: $\limit{W(R)} = \overline{W(R)} \cap V(h_R)$.
In particular, $\limit{W(R)}$ is either empty or an affine variety
of dimension $s-e-1$.
\end{Lemma}

\begin{Lemma}
\label{Lemma:limitWR2}
If $R$ is a primitive regular chain, 
that is, if $R$ is a system of generators of its saturated ideal,
then we have
$\limit{W(R)} = V(R) \cap V(h_R)$.
\end{Lemma}

\begin{Lemma}
\label{Lemma:limitWR3}
If $N$ is a strongly normalized regular chain such that
$\sat{R} = \sat{N}$ and $V(h_N) = V(\widehat{h_R})$ both hold, 
then we have $\limit{W(R)} \subseteq \limit{W(N)}$.
\end{Lemma}

\begin{Lemma}
\label{Lemma:limitWR4}
Let $x \in {\A}^s$ such that $x \not\in W(R)$.
Then $x \in \limit{W(R)}$ holds if and only if
there exists a sequence $({\alpha}_n, n \in {\N})$ 
of points in ${\A}^s$ such that ${\alpha}_n \in W(R)$
for all $n \in {\N}$ and 
$\lim_{n \rightarrow  \infty} \, {\alpha}_n = x$.
\end{Lemma}

\begin{Lemma}
\label{Lemma:limitWR5}
Recall that $R$ writes $\{ r_1, \ldots, r_{e} \}$.
If $e > 1$ holds, 
writing $R' = \{ r_1, \ldots, r_{e-1} \}$ and $r = r_{e}$, 
we have
\begin{equation*}
\limit{W(R' \cup r)} \ \subseteq \ 
\limit{W(R')} \, \cap \,
\limit{W(r)}.
\end{equation*}
\end{Lemma}

\begin{Lemma}
\label{Lemma:NegativePowers}
Let ${\varphi} \in {\C} ( \langle T^{*} \rangle )$
and let $p/q \in {\Q}$ be the order of ${\varphi}$.
Let $({\alpha}_n, n \in {\N})$ be a sequence of complex
numbers converging to zero and let 
$N$ be a positive integer such that 
$({\varphi}({\alpha}_n), n \geq N)$ is well defined.
Then, if $p/q < 0$ holds, the sequence 
$({\varphi}({\alpha}_n), n \geq N)$ escapes to infinity
while if $p/q \geq 0$, the sequence 
$({\varphi}({\alpha}_n), n \geq N)$ converges to 
the complex number ${\varphi}(0)$.
\end{Lemma}

\section{Basic techniques}
\label{sec:limitpoints}

This section is an overview of the basic techniques
of this paper.
This presentation is meant to help the non-expert reader understand
our objectives and solutions.
In particular, the results of this section are stated
for regular chains in three variables, while the statements of
Sections~\ref{sec:theory} to \ref{sec:principle}
do not have this restriction.

Recall that $R \subseteq {\C}[X_1, \ldots, X_s]$
is a regular chain whose saturated ideal has height $1 \leq e \leq s$.
As mentioned in the introduction, we mainly focus
on the case $e = s-1$, that is, \sat{R} has dimension one.

Lemma~\ref{Lemma:limitWR1} and the assumption $e = s-1$ imply that
$\limit{W(R)}$ consists of finitely many points.

We further assume that $R$ is strongly normalized,
thus we have $h_R$ lies in ${\C}[X_1]$.

Lemma~\ref{Lemma:limitWR2}  and the assumption  $h_R \in {\C}[X_1]$ 
imply that computing $\limit{W(R)}$ reduces
to check, for each root $\alpha \in {\C}$ of ${h_R}$ 
whether or not there is a point $x \in \limit{W(R)}$
whose $X_1$-coordinate is $\alpha$.
Without loss of generality, it is enough to develop our
results for the case $\alpha = 0$.
Indeed, a change of coordinates can be used to reduce
to this latter assumption.

We start by considering the case $n=2$.
Thus, our regular chain $R$ consists of a single polynomial 
$r_1 \in {\C}[X_1, X_2]$ whose initial $h_1$ satisfies 
$h_1(0) = 0$.
Lemma~\ref{Lemma:LWRdim1OnePoly} provides a necessary 
and sufficient condition for a point of $({\alpha}, {\beta}) \in {\A}^2$,
with $\alpha = 0$, 
to satisfy $({\alpha}, {\beta}) \in \limit{W( \{ r_1 \} )}$.

Let $d$ be the degree of $r_1$ in $X_2$.
Applying Puiseux Theorem, we consider 
${\varphi}_1, \ldots, {\varphi}_d \in {\C}(\langle X_1^{*} \rangle)$
such that the following holds
\begin{equation}
\label{eq:facto1}
\frac{r_1}{h_1} = (X_2 - {\varphi}_1) \cdots (X_2 - {\varphi}_d)
\end{equation}
in ${\C}(\langle X_1^{*} \rangle) [X_2]$.
We assume that the series 
${\varphi}_1, \ldots, {\varphi}_d $ are numbered in such a way that 
each of ${\varphi}_1, \ldots, {\varphi}_c $ has a non-negative order while
each of ${\varphi}_{c+1}, \ldots, {\varphi}_d $ has a negative order,
for some $c$ such that $0 \leq c \leq d$.

\begin{Lemma}
\label{Lemma:LWRdim1OnePoly}
With $h_1(0)=0$, 
for all  $\beta \in {\C}$, the following two conditions
are equivalent
\begin{enumerate}
\item[$(i)$] $(0,  \beta) \in \limit{W(r_1)}$ holds,
\item[$(ii)$] there exists $1 \leq j \leq c$
               and a sequence $({\alpha}_n, n \in {\N})$
               of complex numbers such that the sequence
                $( {\varphi}_j ({\alpha}_n), n \in {\N})$
               is well defined, we have $h_1({\alpha}_n) \neq 0$
               for all $n \in {\N}$ and we we have
\begin{equation*}
\lim_{n \rightarrow \infty} \, {\alpha}_n = 0 \ \ {\rm and}  \ \
\lim_{n \rightarrow \infty} \, {\varphi}_j ({\alpha}_n) = {\beta}.
\end{equation*}
\end{enumerate}
\end{Lemma}

\begin{proof}
We first prove the implication $(ii) \Rightarrow (i)$.
Equation~(\ref{eq:facto1}) together with $(ii)$
implies $({\alpha}_n, {\varphi}_j ({\alpha}_n)) \in V(r_1)$ for all $n \in {\N}$.
Since we also have
$({\alpha}_n, {\varphi}_j ({\alpha}_n)) \not\in V(h_1)$ for all $n \in {\N}$
and 
$\lim_{n \rightarrow \infty} \, ({\alpha}_n, {\varphi}_j ({\alpha}_n))  = (0, {\beta})$,
we deduce $(i)$, thanks to Lemma~\ref{Lemma:limitWR4}.

We now prove the implication $(i) \Rightarrow (ii)$.
By Lemma~\ref{Lemma:limitWR4}, there exists
a sequence $(({\alpha}_n, {\beta}_n), n \in {\N})$ in ${\A}^2$
such that for all $n \in {\N}$ we have:
          (1) $h_1({\alpha}_n) \neq 0$,
          (2) $r_1 ({\alpha}_n, {\beta}_n) = 0$, and
          (3) $\lim_{n \rightarrow \infty} \, ({\alpha}_n, {\beta}_n) = (0, {\beta})$.
Since  $\lim_{n \rightarrow \infty} \, {\alpha}_n = 0$,
each series 
${\varphi}_1 ({\alpha}_n), \ldots, {\varphi}_d ({\alpha}_n)$
is well defined for $n$ larger than some positive integer $N$.
Hypotheses (1) and (2), together with Equation~(\ref{eq:facto1}),
imply that for all $n \geq N$ the product
\begin{equation*}
({\beta}_n - {\varphi}_1({\alpha}_n)) \cdots
({\beta}_n - {\varphi}_c({\alpha}_n))
({\beta}_n - {\varphi}_{c+1}({\alpha}_n)) \cdots
({\beta}_n - {\varphi}_d({\alpha}_n))
\end{equation*}
is $0$.
Since $\lim_{n \rightarrow \infty} \, {\beta}_n = {\beta}$,
and by definition of the integer $c$,
each of the sequences $({\beta}_n - {\varphi}_1({\alpha}_n)), \ldots, 
({\beta}_n - {\varphi}_c({\alpha}_n))$ converges
while each of the sequences $({\beta}_n - {\varphi}_{c+1}({\alpha}_n)), \ldots, 
({\beta}_n - {\varphi}_d({\alpha}_n))$ escapes to infinity.
Thus, for $n$ large enough
the product 
$({\beta}_n - {\varphi}_1({\alpha}_n)) \cdots
({\beta}_n - {\varphi}_c({\alpha}_n))$ is zero.
Therefore, one of sequences $({\beta}_n - {\varphi}_1({\alpha}_n)), \ldots, 
({\beta}_n - {\varphi}_c({\alpha}_n))$ converges to $0$
and the conclusion follows.
\end{proof}

Lemmas~\ref{Lemma:NegativePowers} and \ref{Lemma:LWRdim1OnePoly}
immediately imply the following.

\begin{Proposition}
\label{Propo:LWRdim1OnePoly1}
With $h_1(0) = 0$, for all  $\beta \in {\C}$, we have
\begin{equation*}
(0,  \beta) \in \limit{W(r_1)} \ \ \iff \ \ 
\beta \in \{ {\varphi}_1(0), \ldots, {\varphi}_c(0) \}.
\end{equation*}
\end{Proposition}

Next, we consider the case $n=3$.
Hence, our regular chain $R$ consists of two polynomials
$r_1 \in {\C}[X_1, X_2]$ and $r_2 \in {\C}[X_1, X_2, X_3]$
with respective initials $h_1$ and $h_2$.
We assume that $0$ is a root of the product $h_1 h_2$
and we are looking for all ${\beta} \in {\C}$ and 
all ${\gamma} \in {\C}$ such that 
$(0, {\beta}, {\gamma}) \in \limit{W(r_1, r_2)}$.

Lemma~\ref{Lemma:limitWR5} tells us that 
$(0, {\beta}, {\gamma}) \in \limit{W(r_1, r_2)}$
implies 
$(0, {\beta}) \in \limit{W(r_1)}$.
This observation together with Proposition~\ref{Propo:LWRdim1OnePoly1}
yields immediately the following.

\begin{Proposition}
\label{Propo:LWRdim1OTowPoly1}
With $h_1(0) = 0$ and $h_2(0) \neq 0$, assuming that $r_1$ is primitive
over ${\C}[X_1]$, 
for all ${\beta} \in {\C}$ and 
all ${\gamma} \in {\C}$, we have 
\begin{equation*}
(0,  \beta, \gamma) \in \limit{W(r_1, r_2)} \ \ \iff \ \ 
(0,  \beta, \gamma) \in V(r_1, r_2).
\end{equation*}
\end{Proposition}

We turn now our attention to the case $h_1(0) = h_2(0) = 0$.
Since $(0, {\beta}) \in \limit{W(r_1)}$ is a necessary
condition for $(0, {\beta}, {\gamma}) \in \limit{W(r_1, r_2)}$
to hold we apply Proposition~\ref{Propo:LWRdim1OnePoly1}
and assume 
$\beta \in \{ {\varphi}_1(0), \ldots, {\varphi}_c(0) \}$.
Without loss of generality, we further assume $\beta = 0$.
For each $1 \leq j \leq c$, such that 
 ${\varphi}_j(0) = 0$ holds,
we define the univariate polynomial $f^j_2 \in {\C}(\langle X_1^{*} \rangle)[X_3]$ by
\begin{equation}
\label{eq:f2}
f^j_2(X_1,X_3) = r_2(X_1, {\varphi}_j(X_1), X_3).
\end{equation}
Let $b$ be the degree of $f^j_2$.
Applying again Puiseux theorem, 
we consider ${\psi}_1, \ldots, {\psi}_b \in {\C}(\langle X_1^{*} \rangle)$
such that the following holds
\begin{equation}
\label{eq:facto2}
\frac{f_2^j}{h_2} = (X_3 - {\psi}_1) \cdots (X_3 - {\psi}_b)
\end{equation}
in ${\C}(\langle X_1^{*} \rangle) [X_3]$.
We assume that the series 
${\psi}_1, \ldots, {\psi}_b $ are numbered in such a way that 
each of ${\psi}_1, \ldots, {\psi}_a $ has a non-negative order while
each of ${\psi}_{a+1}, \ldots, {\psi}_b $ has a negative order,
for some $a$ such that $0 \leq a \leq b$.

\begin{Lemma}
\label{Lemma:LWRdim1TwoPolys}
For all  $\gamma \in {\C}$, the following two conditions
are equivalent.
\begin{enumerate}
\item[$(i)$] $(0, 0, \gamma) \in \limit{W(r_1, r_2)}$ holds,
\item[$(ii)$] there exist integers $j,k$ with
              $1 \leq j \leq c$ and $1 \leq k \leq a$,
              and two sequences
              $({\alpha}_n, n \in {\N})$, 
              $({\beta}_n, n \in {\N})$ of complex numbers 
              such that:
\begin{enumerate}
\item[$(a)$] the sequences $( {\varphi}_j ({\alpha}_n), n \in {\N})$
             and $( {\psi}_k ({\beta}_n), n \in {\N})$ are well defined,
\item[$(b)$] $h_1({\alpha}_n) \neq 0$ and $h_2({\alpha}_n) \neq 0$, 
             for all $n \in {\N}$,
\item[$(c)$] ${\beta}_n = {\varphi}_j ({\alpha}_n)$, for all $n \in {\N}$,
\item[$(d)$] $\lim_{n \rightarrow \infty} \, 
             ({\alpha}_n, {\beta}_n, {\psi}_k ({\beta}_n)) = (0, 0, \gamma)$.
\end{enumerate}
\end{enumerate}
\end{Lemma}

\begin{proof}
Proving the implication $(ii) \Rightarrow (i)$ is easy.
We now prove the implication $(i) \Rightarrow (ii)$.
By Lemma~\ref{Lemma:limitWR4}, there exists
a sequence $(({\alpha}_n, {\beta}_n, {\gamma}_n), n \in {\N})$ in ${\A}^3$
s.t. for all $n \in {\N}$ we have:
          (1) $h_1({\alpha}_n) \neq 0$,
          (2) $h_2({\alpha}_n) \neq 0$, 
          (3) $r_1 ({\alpha}_n, {\beta}_n) = 0$, 
          (4) $r_2 ({\alpha}_n, {\beta}_n, {\gamma}_n) = 0$, 
(5) $\lim_{n \rightarrow \infty} \, ({\alpha}_n, {\beta}_n, {\gamma}_n) = (0, 0, {\gamma})$.
Following the proof of Lemma~\ref{Lemma:LWRdim1OnePoly}, 
we know that for $n$ large enough
the product 
$({\beta}_n - {\varphi}_1({\alpha}_n)) \cdots
({\beta}_n - {\varphi}_c({\alpha}_n))$ is zero.
Therefore, from one of the sequences
$({\beta}_n - {\varphi}_1({\alpha}_n)), \ldots, 
({\beta}_n - {\varphi}_c({\alpha}_n))$,
say the $j$-th, one can extract an (infinite) sub-sequence
whose terms are all zero.
Thus, without loss of generality, we assume that
${\beta}_n = {\varphi}_j ({\alpha}_n)$ holds, for all $n \in {\N}$.
Hence, for all $n \in {\N}$, we have
$f^j_2({\alpha}_n,{\gamma}_n) = r_2({\alpha}_n, {\beta}_n, {\gamma}_n) = 0$.
Together with Equation~(\ref{eq:facto2})
and following the proof of Lemma~\ref{Lemma:LWRdim1OnePoly}, 
we deduce the desired result.
\end{proof}

Lemmas~\ref{Lemma:NegativePowers} and \ref{Lemma:LWRdim1TwoPolys}
immediately imply the following.

\begin{Proposition}
\label{Propo:LWRdim1OneTwoPolys}
For all  $\gamma \in {\C}$, the following two conditions
are equivalent.
\begin{enumerate}
\item[$(i)$] $(0, 0, \gamma) \in \limit{W(r_1, r_2)}$ holds,
\item[$(ii)$] there exist integers $j,k$ with
              $1 \leq j \leq c$ and $1 \leq k \leq a$, such that 
               ${\varphi}_j (0) = 0$ and ${\psi}_k (0) = \gamma$.
\end{enumerate}
\end{Proposition}

Therefore, applying  Puiseux theorem to $r_1$ and $f^j_2$,
then checking the constant terms of the series
${\psi}_1, \ldots, {\psi}_b $ provides a way to compute
all $\gamma \in {\C}$ such that $(0, 0, \gamma)$
is a limit point of $W(r_1, r_2)$.
Theorem~\ref{thrm:PuiseuxParamRC} in Sections~\ref{sec:theory}
states this principle formally
for an arbitrary regular chain $R$ in dimension one.

Finally, one should also consider the case 
$h_1(0) \neq 0, h_2(0) = 0$.
In fact, it is easy to see that this latter case can be handled
in a similar manner as the case $h_1(0) = 0, h_2(0) = 0$.

\section{Puiseux expansions of a regular chain}
\label{sec:theory}
In this section, we introduce the notion of Puiseux expansions
of a regular chain, motivated by the work of~\cite{MaurerJoseph80,Mo92} on 
Puiseux expansions of space curves.

\begin{Lemma}
\label{Lemma:V0}
Let $R=\{r_1,\ldots,r_{s-1}\}\subset\C[X_1<\cdots<X_s]$
be a strongly normalized regular chain whose saturated ideal has 
dimension one.
Recall that $h_R(X_1)$ denotes the product of the initials of polynomials in $R$.
Let $\rho>0$ be small enough such that the set $0<\abs{X_1}<\rho$
does not contain any zeros of $h_R$.
Denote by $U_{\rho} := \{x=(x_1,\ldots,x_s)\in\C^s \mid 0<\abs{x_1}<\rho \}$.
Denote by $V_{\rho}(R) := V(R)\cap U_{\rho}$.
Then we have $W(R)\cap U_{\rho}=V_{\rho}(R)$.
\iflongversion
Let $R' := \{ \primpart{r_1},\ldots,\primpart{r_{s-1}}\}$. 
Then $V_{\rho}(R)=V_{\rho}(R')$.
\fi
\end{Lemma}
\begin{proof}
Let $x\in W(R)\cap U_{\rho}$, then $x\in V(R)$ and $x\in U_{\rho}$ hold, 
which implies that $W(R)\cap U_{\rho}\subseteq V(R)\cap U_{\rho}$.
Let $x\in V(R)\cap U_{\rho}$. 
Since $U_{\rho}\cap V(h_R)=\emptyset$, 
we have $x\in W(R)$. Thus $V(R)\cap U_{\rho}\subseteq W(R)\cap U_{\rho}$.
So $W(R)\cap U_{\rho}=V_{\rho}(R)$. 
\iflongversion
Similarly we have $V_{\rho}(R)=V_{\rho}(R')$.
\fi
\end{proof}

\begin{Notation}
Let $W \subseteq \C^{s}$.
Denote ${\rm lim}_{0}(W):=\{x=(x_1,\ldots,x_s)\in\C^s\mid x\in \limit{W}
        \ {\rm and} \ x_1=0\}$.
\end{Notation}

\begin{Lemma}
\label{Lemma:lim0}
Let $R=\{r_1,\ldots,r_{s-1}\}\subset\C[X_1<\cdots<X_s]$.
Then we have ${\rm lim}_{0}(W(R))={\rm lim}_{0}(V_{\rho}(R))$.
\end{Lemma}
\begin{proof}
By Lemma~\ref{Lemma:V0}, we have $W(R)\cap U_{\rho}(R)=V_{\rho}(R)$.
On the other hand   ${\rm lim}_{0}(W(R))={\rm lim}_{0}(W(R)\cap U_{\rho}(R))$.
Thus ${\rm lim}_{0}(W(R))={\rm lim}_{0}(V_{\rho}(R))$ holds.
\end{proof}

\begin{Lemma}
\label{Lemma:zero-series}
Let $R$ be as in Lemma~\ref{Lemma:V0}.
For $1\leq i\leq s-1$, let $d_i := \deg(r_i, X_{i+1})$.
Then $R$ generates a zero-dimensional ideal in $\C(\langle X_1^*\rangle)[X_2,\ldots,X_s]$.
Let $V^{*}(R)$ be the zero set of $R$ in $\C(\langle X_1^*\rangle)^{s-1}$.
Then $V^{*}(R)$ has exactly $\prod_{i=1}^{s-1} d_i$  points, counting multiplicities.
\end{Lemma}
\begin{proof}
It follows directly from the definition of regular chain, 
Bezout bound and the fact that $\C(\langle X_1^*\rangle)^{s-1}$
is an algebraically closed field.
\end{proof}

\begin{Definition}
\label{Definition:PuiseuxExpansionOfRC}
We use the notions in Lemma~\ref{Lemma:zero-series}.
Each point in $V^{*}(R)$ is called a {\em Puiseux expansion} of $R$.
\end{Definition}

\begin{Notation}
Let $m=\abs{V^*(R)}$.
Write $V^*(R)= \{\Phi_1, \ldots, \Phi_m \}$
with $\Phi_i=(\Phi_i^1(X_1),\ldots,\Phi_i^{s-1}(X_1))$,
for $i=1,\ldots,m$.
Let $\rho>0$ be small enough such that 
for $1 \leq i \leq m$, $ 1 \leq j \leq s-1$, each $\Phi_i^j(X_1)$
converges in $0<\abs{X_1}<\rho$.
We define
$
V^*_{\rho}(R):=\cup_{i=1}^m \{x\in\C^s\mid 0<\abs{x_1}<\rho, x_{j+1}=\Phi_i^j(x_1), j=1,\ldots,s-1\}.
$

\end{Notation}

\begin{Theorem}
\label{Theorem:VStarPho}
We have $V^*_{\rho}(R)=V_{\rho}(R)$.
\end{Theorem}
\begin{proof}
We prove this by induction on $s$.
For $i=1,\ldots,s-1$, recall that $h_i$  is the initial of $r_i$.
If $s=2$, we have 
$$
r_1(X_1,X_2)=h_1(X_1)\prod_{i=1}^{d_1}(X_2-\Phi_i^1(X_1)).
$$
So $V^*_{\rho}(R)=V_{\rho}(R)$ clearly holds.

Write $R=R'\cup\{r_{s-1}\}$,
      $X' = X_2, \ldots, X_{s-1}$,
      $X=(X_1,X',X_{s})$,
      $x' = x_2, \ldots, x_{s-1}$,
      $x=(x_1,x',x_{s})$, 
       and $m'=\abs{V^*(R')}$.
For $i=1,\ldots,m$, let $\Phi_i=(\Phi_i', \Phi_i^{s-1})$,
where $\Phi_i'$ stands for $\Phi_i^1, \ldots, \Phi_i^{s-2}$.
Assume the theorem holds for $R'$, 
that is $V^*_{\rho}(R')=V_{\rho}(R')$.
For any  $i=1,\ldots,m'$, there exist $i_k\in\{1,\ldots,m\}$, $k=1,\ldots,d_{s-1}$
such that
\begin{equation}
\label{equation:expansion}
r_{s-1}(X_1,X'=\Phi_i', X_{s})=h_1(X_1)\prod_{k=1}^{d_{s-1}}(X_s-\Phi_{i_k}^{s-1}(X_1)).
\end{equation}
Note that $V^*(R)=\cup_{i=1}^{m'}\cup_{k=1}^{d_{s-1}}\{ (X'=\Phi_i', X_s=\Phi_{i_k}^{s-1})\}$.
Therefore, by induction hypothesis and Equation~(\ref{equation:expansion}), 
we have
$$
\begin{array}{rcl}
V^*_{\rho}(R) &=&\cup_{i=1}^{m'}\cup_{k=1}^{d_{s-1}}\{x\mid x\in U_{\rho}, x'=\Phi_i'(x_1), x_s=\Phi_{i_k}^{s-1}(x_1)\}\\
            &=&\cup_{k=1}^{d_{s-1}}\{x\mid (x_1,x')\in V^*_{\rho}(R'), x_s=\Phi_{i_k}^{s-1}(x_1)\}\\
            &=&\{x\mid (x_1,x')\in V^*_{\rho}(R'), r_{s-1}(x_1,x',x_s)=0\}\\
            &=&\{x\mid (x_1,x')\in V_{\rho}(R'), r_{s-1}(x_1,x',x_s)=0\}\\
            &=&V_{\rho}(R).           
\end{array}
$$
\end{proof}

\begin{Theorem}
\label{Theorem:limit}
Let $V^*_{\geq 0}(R):=\{\Phi=(\Phi^1,\ldots,\Phi^{s-1})\in V^*(R)\mid \ord{\Phi^j}\geq 0, j=1,\ldots,s-1 \}$.
Then we have 
$$
{\rm lim}_0(W(R))=\cup_{\Phi\in V^*_{\geq 0}(R)} \{(X_1=0, \Phi(X_1=0))\}.
$$
\end{Theorem}
\begin{proof}
By definition of $V^*_{\geq 0}(R)$, 
we immediately have 
$$
{\rm lim}_0(V^*_{\rho}(R))=\cup_{\Phi\in V^*_{\geq 0}(R)} \{(X_1=0, \Phi(X_1=0))\}.
$$
Next, by Theorem~\ref{Theorem:VStarPho}, 
we have $V^*_{\rho}(R)=V_{\rho}(R)$.
Thus, we have ${\rm lim}_0(V^*_{\rho}(R))={\rm lim}_0(V_{\rho}(R))$.
Besides, 
with Lemma~\ref{Lemma:lim0}, we have ${\rm lim}_{0}(W(R))={\rm lim}_{0}(V_{\rho}(R))$.
Thus the theorem holds.
\end{proof}

\begin{Definition}
\label{Definition:PuiseuxParamRC}
Let $V^*_{\geq 0}(R):=\{\Phi=(\Phi^1,\ldots,\Phi^{s-1})\in V^*(R)\mid \ord{\Phi^j}\geq 0, j=1,\ldots,s-1 \}$.
Let $M=\abs{V^*_{\geq 0}(R)}$.
For each $\Phi_i=(\Phi_i^1,\ldots,\Phi_i^{s-1})\in V^*_{\geq 0}(R)$, $1\leq i\leq M$,
we know that $\Phi_i^j\in\C(\langle X_1^*\rangle)$.
Moreover, by Equation~(\ref{equation:expansion}), 
we know that for $j=1,\ldots,s-1$, $\Phi_i^j$ is a Puiseux expansion
of $r_{j}(X_1,X_2=\Phi_i^1,\ldots,X_j=\Phi_i^{j-1} , X_{j+1})$.
Let $\varsigma_{i,j}$ be the ramification index of $\Phi_i^j$
and $(T^{\varsigma_{i,j}}, X_{j+1}=\varphi_i^j(T))$, where $\varphi_i^j\in\C\langle T\rangle$,
be the corresponding Puiseux parametrization of $\Phi_i^j$.
Let $\varsigma_i$ be the least common multiple of $\{\varsigma_{i,1},\ldots,\varsigma_{i,s-1}\}$.
Let $g_{i}^{j}=\varphi_i^j(T=T^{\varsigma_i/\varsigma_{i,j}})$.
We call the set $\frak{G}_R := \{(X_1=T^{\varsigma_{i}}, X_2=g_i^1(T),\ldots,X_{s}=g_i^{s-1}(T)), i=1,\ldots,M\}$
a {\em system of Puiseux parametrizations} of $R$.
%
\end{Definition}

\begin{Theorem}
\label{thrm:PuiseuxParamRC}
We have
$$
{\rm lim}_0(W(R))=\frak{G}_R(T=0).
$$
\end{Theorem}
\begin{proof}
It follows directly from Theorem~\ref{Theorem:limit} and Definition~\ref{Definition:PuiseuxParamRC}.
\end{proof}

\section{Puiseux parametrization in finite accuracy}
\label{sec:newton}

In this section, we define the Puiseux parametrizations 
of a polynomial $f\in\C\langle X\rangle[Y]$ in finite accuracy, 
see Definition~\ref{Definition:PuiseuxParaFinite}.

For $f\in\C\langle X\rangle[Y]$, 
we define the approximation $\widetilde{f}$ of $f$
for a given finite accuracy, see Definition~\ref{Definition:approximation}.
This approximation $\widetilde{f}$ of $f$ is a polynomial in $\C[X, Y]$.
In Section~\ref{sec:bounds}, we prove that 
in order to compute a Puiseux parametrizations of $f$ 
of a given accuracy,
it suffices to compute a Puiseux parametrization of $\widetilde{f}$
of some finite accuracy.

In this section, we review and adapt  the classical Newton-Puiseux 
algorithm to compute Puiseux parametrizations
of a polynomial $f\in\C[X, Y]$ of a given accuracy.
Since we do not need to compute the singular
part of Puiseux parametrizations, 
the usual requirement $\discrim{f, Y}\neq 0$
is dropped.


\begin{Definition}
\label{Definition:approximation}
  Let $f=\sum_{i=0}^{\infty} a_i X^i \in \C[[X]]$.  
  For any $\tau\in\N$, 
  let $f^{(\tau)} := \sum_{i=0}^{\tau} a_iX^i$.  We call
  $f^{(\tau)}$ the {\em polynomial part} of $f$ %
{\em of accuracy} $\tau+1$.
Now, let  $f=\sum_{i=0}^{d} a_i(X) Y^i\in\C\langle X\rangle[Y]$.
For any $\tau\in\N$, 
we call $\widetilde{f}^{(\tau)} := \sum_{i=0}^d a_i^{(\tau)} Y^i$
the {\em approximation of $f$ of accuracy $\tau+1$}.
\end{Definition}

\begin{Definition}
\label{Definition:PuiseuxParaFinite}
Let $f\in\C\langle X\rangle[Y]$, $\deg(f, Y)>0$.
Let $\sigma, \tau\in{\N}_{>0}$ and $g(T)=\sum_{k=0}^{\tau-1}b_k T^k$.
The pair $(T^{\sigma}, g(T))$ is called 
{\em a Puiseux parametrization of $f$ of accuracy $\tau$} if 
there exists an irreducible Puiseux parametrization $(T^{\varsigma}, \varphi(T))$
of $f$ such that
\begin{itemizeshort}
\item[$(i)$] $\sigma$ divides  $\varsigma$.
\item[$(ii)$] $\gcd{\sigma, b_0,\ldots,b_{\tau-1}}=1$.
\item[$(iii)$] $g(T^{\varsigma/\sigma})$ is the polynomial part of $\varphi(T)$
of accuracy $(\varsigma/\sigma)(\tau-1)+1$.
\end{itemizeshort}
Note that if $\sigma=\varsigma$, then 
$g(T)$ is simply the polynomial part of $\varphi(T)$ of accuracy $\tau$.
\end{Definition}

We borrow the following notion from~\cite{Duv89}
in order to state an algorithm 
for computing Puiseux parametrizations.

\begin{Definition}
A $\C$-term\footnote{It is a simplified version of Duval's definition.} 
is defined as a triple $t=(q, p,\beta)$, where $q$ and $p$ 
are coprime integers, $q>0$ and $\beta \in {\C}$ is non-zero.
A $\C$-expansion is a sequence $\pi=(t_1,t_2,\ldots)$ of $\C$-terms,
where $t_i=(q_i, p_i,\beta_i)$, 
We say that  $\pi$ is {\em finite} if there are only finitely many elements in $\pi$.
\end{Definition}

\begin{Definition}
\label{Definition:ExpansionParametrization}
Let $\pi=(t_1,\ldots,t_N)$ be a finite $\C$-expansion.
We define a pair $(T^{\sigma}, g(T))$ of polynomials in ${\C}[T]$ in the following manner:
\begin{itemizeshort}
\item if $N=1$, set $\sigma=1$, $g(T)=0$ and $\delta_N = 0$,
\item otherwise, let $a := \prod_{i=1}^N q_i$, 
$c_i := \sum_{j=1}^i\left( p_j\prod_{k=j+1}^Nq_k \right)$ ($1\leq i\leq N$), 
and $\delta_i := c_i/\gcd{a, c_1,\ldots,c_N}$ ($1\leq i\leq N$).
Set $\sigma := a/\gcd{a, c_1,\ldots,c_N}$ and $g(T) := \sum_{i=1}^N \beta_i T^{\delta_i}$.
\end{itemizeshort}
We call the pair $(T^{\sigma}, g(T))$ {\em the corresponding Puiseux para-metrization of $\pi$ of accuracy $\delta_N+1$}.
Denote by {\sf ConstructParametrization} an algorithm to compute $(T^{\sigma}, g(T))$
from $\pi$.
\end{Definition}

\begin{Definition}
Let $f\in\C\langle X\rangle[Y]$ and 
write $f$ as $f(X, Y) := \sum_{i=0}^d \left(\sum_{j=0}^{\infty} a_{i, j}X^j \right) Y^i$.
The {\em Newton Polygon} of $f$
is defined as the lower part of the convex hull 
of the set of points $(i,j)$
in the plane such that $a_{i,j}\neq 0$.
\end{Definition}

\iflongversion

Let $f\in\C\langle X\rangle[Y]$.
Next we present an algorithm, called {\sf NewtonPolygon} to compute 
the segments in the Newton Polygon of $f$.
This algorithm is from R.J. Walker's book~\cite{RW78}.

\begin{itemizeshort}
\item[] $\NewtonPolygon{f, I}$
\item[] {\bf Input:} A polynomial $f\in\C\langle X\rangle[Y]$; 
a controlling flag $I$, whose value is $1$ or $2$.
\item[] {\bf Output:} The Newton Polygon of $f$.
If $I=1$, only segments with non-positive slopes
are computed.
If $I=2$, only segments with negative slopes are computed.

\item[] {\bf Description:}
\begin{itemizeshort}
\item Write $f$ as $f=\sum_{i=0}^d b_{i}(X)Y^i$, where $b_i(X)=\sum_{j=0}^{\infty} a_{i,j}X^j$.
\item For $0\leq i\leq d$, define $\delta_i := \ord{b_i}$.
\item For $0\leq i\leq d$, we plot the points $P_i$ with coordinates $(i,\delta_i)$;
      we omit $P_i$ if $\delta_i=\infty$.
\item We join $P_0$ to $P_d$ with a convex polygonal arc each of whose vertices
is a $P_i$ and such that no $P_i$ lies below the arc.
\item If $I=1$, output all segments with non-positive slopes in the polygon;
      if $I=2$, output all segments with negative slopes in the polygon.
\end{itemizeshort}
\end{itemizeshort}
\fi


Next we present the specification 
of several other sub-algorithms 
which are necessary to present Algorithm~\ref{alg:newton}
for computing Puiseux parametrization of some finite accuracy
as defined in Definition~\ref{Definition:PuiseuxParaFinite}.
\smallskip

\begin{itemizeshort}
\item[] $\NewPolynomial{f, t, \ell}$
\item[] {\bf Input:} $f\in\C[X, Y]$; a $\C$-term $t=(q, p, \beta)$; $\ell\in\N$.
\item[] {\bf Output:} A polynomial 
$X^{-\ell}f(X^q, X^p(\beta + Y))\in\C[X, Y]$.
\end{itemizeshort}
\smallskip

\begin{itemizeshort}
\item[] $\SegmentPoly{f, \Delta}$
\item[] {\bf Input:}  $f\in\C[X, Y]$; $\Delta$ is a segment 
of the Newton Polygon of $f$.
\item[] {\bf Output:} A quadruple $(q, p, \ell, \phi)$
such that the following holds
\begin{itemizeshort}
\item $q, p,\ell\in\N$; $\phi\in\C[Z]$; $q$ and $p$ are coprime, $q>0$.
\item For any $(i,j)\in\Delta$, we have $qj+pi=\ell$.
\item Let $i_0:=\min{\{i\mid (i,j)\in\Delta\}}$, 
we have 
$\phi=\sum_{(i,j)\in\Delta}a_{i,j}Z^{(i-i_0)/q}.$
\end{itemizeshort}
\end{itemizeshort}
\smallskip

\begin{Theorem}
Algorithm~\ref{alg:newton} terminates and is correct.
\end{Theorem}
\begin{proof}
It directly follows from the proof of Newton-Puiseux 
algorithm in Walker's book~\cite{RW78}, 
the relation between $\C$-expansion
and Puiseux parametrization discussed
in Duval's paper~\cite{Duv89},
and Definitions~\ref{Definition:ExpansionParametrization}
and~\ref{Definition:PuiseuxParaFinite}.
\end{proof}

\begin{algorithm}
\linesnumbered
\caption{$\NonzeroTerm{f, I}$}
\KwIn{
 $f\in\C[X, Y]$; $I=1$ or $2$
}
\KwOut{
A finite set of pairs $(t, \ell)$, where $t$ is a $\C$-term, 
and $\ell\in\N$.
}
\Begin{
$S := \emptyset$\;
\For{each $\Delta\in\NewtonPolygon{f, I}$}{
  $(q, p,\ell, \phi) := \SegmentPoly{f,\Delta}$\;
    \For{each root $\xi$ of $\phi$ in $\C$}{
        \For{each root $\beta$ of $U^q-\xi$ in $\C$}{
             $t := (q,p,\beta)$\;
             $S := S\cup \{ (t, \ell) \}$
        }

    }
}

}
\end{algorithm}

\begin{algorithm}
\label{alg:newton}
\linesnumbered
\caption{${\sf NewtonPuiseux}$}
\KwIn{$f\in\C[X, Y]$; a given accuracy $\tau>0\in\N$.
}
\KwOut{
All the Puiseux parametrizations of $f$
of accuracy $\tau$.
}
\Begin{
   $\pi := (~)$;
   $S := \{(\pi, f)\}$\;
   \While{$S\neq \emptyset$}{
     choose $(\pi^*, f^*)\in S$; $S := S\setminus \{(\pi^*, f^*)\}$\;
     {\bf if} $\pi^*=(~)$ {\bf then} $I := 1$ {\bf else} $I := 2$\;
     $(T^{\sigma}, g(T)) := \ConstructParametrization{\pi^*}$\;
     \If{$\deg(g(T), T)+1<\tau$}{
       $C := \NonzeroTerm{f^*, I}$\;
       \uIf{$C=\emptyset$}{
           output $(T^{\sigma}, g(T))$ \tcp{a finite Puiseux parametrization is found}
       }
       \Else{
         \For{each $(t=(p, q, \beta), \ell)\in C$}{
           $\pi^{**} :=\pi^*\cup(t)$\;
           $f^{**} := \NewPolynomial{f^*, t, \ell}$\;
           $S := S\cup \{ (\pi^{**}, f^{**})\}$
         }
       }
     }
     output $(T^{\sigma}, g(T))$
   }
}
\end{algorithm}

\section{Computing in finite accuracy}
\label{sec:approximation}

Let $f\in\C\langle X\rangle[Y]$.
In this section, we consider the following problems.
\begin{itemizeshort}
\item[$(a)$] Is it possible to use an approximation of $f$ of some finite accuracy $m$
in order to compute a Puiseux parametrization of $f$ 
of some finite accuracy $\tau$?
\item[$(b)$] If yes,  how to deduce $m$ from $f$ and $\tau$?
\item[$(c)$] Provide a bound on $m$.
\end{itemizeshort}
Theorem~\ref{Theorem:finite-bound} provides the answers 
to $(a)$ and $(b)$ while 
Lemma~\ref{Lemma:generic} answers $(c)$.

\iflongversion
\begin{Lemma}[\cite{GF01}]
  \label{Lemma:GFsubs}
Let $\underline{X}=X_1,\ldots,X_s$ and $\underline{Y}=Y_1,\ldots,Y_m$.
  For $g_1, \ldots, g_s \in {\C}[[ \underline{Y} ]]$, with ${\rm
    ord}(g_i) \geq 1$, there is a ${\C}$-algebra homomorphism (called
  the {\em substitution homomorphism})
  \begin{center}
    $ {\Phi}_g:
    \begin{array}{rcl} {\C}[[ \underline{X} ]] & \longrightarrow & {\C}[[ \underline{Y} ]] \\
      f & \longmapsto & f(g_1(\underline{Y}),\ldots,
      g_s(\underline{Y})).
    \end{array}
    $
  \end{center}
  Moreover, if $g_1, \ldots, g_s $ are convergent power series, then
  we have ${\Phi}_g({\C} \langle \underline{X} \rangle) \subseteq {\C}
  \langle \underline{Y} \rangle$ holds.
\end{Lemma}
\fi

\begin{Definition}[\cite{GF01}]
Let $f=\sum a_{\mu\nu}X^{\mu}Y^{\nu}\in\C[[X, Y]]$.
The {\em carrier} of $f$ is defined as 
$$
{\rm carr}(f)=\{(\mu,\nu)\in\N^2 \ \mid \  a_{\mu\nu}\neq 0\}.
$$
\end{Definition}

\begin{Lemma}
\label{Lemma:one-iteration}
Let $f\in \C\langle X \rangle[Y]$.
Let $d := \deg(f, Y)>0$.
Let $q\in {\N}_{>0}$,  $p, \ell\in\N$ and assume that $q$ and $p$ are coprime.
Let $\beta\neq 0 \in\C$.
Assume that $q, p, \ell$ define a line $ L : qj + pi = \ell$ in $(i, j)$ plane
such that 
\begin{itemizeshort}
\item[$(a)$] There are at least two points $(j_1, i_1)\in\carr{f}$ 
and $(j_2, i_2)\in \carr{f}$
on $L$ with $i_1\neq i_2$.
\item[$(b)$] For any $(j, i)\in\carr{f}$, we have $qj + pi \geq \ell$.
\end{itemizeshort}
Let $f_1 := X_1^{-\ell}f(X_1^q, X_1^p(\beta + Y_1))$.
Then, we have the following results
\begin{itemizeshort}
\item[$(i)$] We have $f_1\in\C\langle X_1 \rangle[Y_1]$.
\item[$(ii)$] For any given $m_1\in\N$, there exists a finite number $m\in\N$ 
such that the approximation of $f_1$
of accuracy $m_1$ can be computed 
from the approximation of $f$ of accuracy $m$.
\item[$(iii)$] Moreover, it suffices to take $m=\lfloor \frac{m_1+\ell}{q} \rfloor$.
\end{itemizeshort}
\end{Lemma}

\begin{proof}
\iflongversion
Since $q >0$ holds, we know that 
$\ord{X_1^q}=q>0$ holds.
We also have $f(X_1^q, X_1^p(\beta + Y_1))\in\C\langle X_1 \rangle[Y_1]$.
Let $f(X, Y) := \sum_{i=0}^d \left(\sum_{j=0}^{\infty} a_{i, j}X^j \right) Y^i$.
Then we have
$
f_1(X_1, Y_1)
             = \sum_{i=0}^d \left(\sum_{j=0}^{\infty} a_{i, j}  X_1^{(qj+pi-\ell)} \right) (\beta + Y_1)^i. 
$
Since for any $(j, i)\in\carr{f}$, we have $qj + pi \geq \ell$, 
the power of $X_1$ cannot be negative. 
By Lemma~\ref{Lemma:GFsubs}, we have $f_1\in\C\langle X_1 \rangle[Y_1]$. 
That is $(i)$ holds.

We prove $(ii)$. We have
$$
\begin{array}{rcl}
&&f_1(X_1, Y_1)~\mbox{mod}~\langle X_1^{m_1} \rangle \\
&=& \sum_{i=0}^d \left(\sum_{qj+pi-\ell < m_1} a_{i, j}  X_1^{(qj+pi-\ell)} \right) (\beta + Y_1)^i.
\end{array}
$$
Since $q\in {\N}_{>0}$ and $m_1$, $\ell$ and $i$ are all finite, 
we know that $j$ has to be finite. 
In other words, there exists a finite $m$ 
such that 
the approximation of $f_1$ of accuracy $m_1$
can be computed from 
the approximation of $f$ of accuracy $m$.
That is, $(ii)$ holds.

Since the first $m_1$ terms of $f_1$
depends on the $j$-th terms  of $f$, which satisfies 
the constraint $qj+pi-\ell < m_1$, 
we have
$
j < \frac{(m_1+\ell)-pi}{q} \leq \frac{(m_1+\ell)}{q}.
$
Let $m'$ be the the maximum of these $j$'s . 
Now we have
$
m'-1 < \frac{(m_1+\ell)}{q}.
$
Since $m'$ is an integer, 
we have 
$
m'\leq \lfloor \frac{(m_1+\ell)}{q}\rfloor
$
holds.
Let $m=\lfloor \frac{(m_1+\ell)}{q}\rfloor$.
Next we show shat $m_1\geq 1$ implies that $m\geq 1$ holds.
If there is at least one point $(i,j)\in L$
such that $j\geq 1$, then we have $\ell\geq q$, 
which implies $m\geq 1$.
If the $j$-coordinates of all points on $L$ is $0$, 
then $q=1$ and $\ell=0$, 
which implies also $m\geq 1$.
Thus $(iii)$ is proved.
\fi
\end{proof}

\iflongversion

\begin{Remark}
\label{Remark:monotone}
We use the same notations as in the previous Lemma.
In particular, let $f(X, Y) := \sum_{i=0}^d \left(\sum_{j=0}^{\infty} a_{i, j}X^j \right) Y^i$
and $f_1 := X_1^{-\ell}f(X_1^q, X_1^p(\beta + Y_1))$.
For a fixed term $a_{i,j}X^jY^i$ of $f$, 
it appears in $f_1$ as 
$$
a_{i,j}X_1^{qj+pi-\ell}(\beta+Y_1)^i
= \sum_{k=0}^i \left({i \choose k}\beta^{i-k} a_{i,j}X_1^{qj+pi-\ell}   \right) Y_1^{k}  .
$$
For two fixed terms $a_{i,j_1}X^{j_1}Y^i$ and $a_{i,j_2}X^{j_2}Y^i$ of $f$ with $j_1<j_2$, 
\ifcomment
they respectively appear in $f_1$ as 
$$
a_{i,j_1}X_1^{q{j_1}+pi-\ell}(\beta+Y_1)^i
=\sum_{k=0}^i \left({i \choose k}\beta^{i-k}  a_{i,j_1}X_1^{q{j_1}+pi-\ell} \right) Y_1^{k}
$$
and 
$$
a_{i,j_2}X_1^{q{j_2}+pi-\ell}(\beta+Y_1)^i
=\sum_{k=0}^i \left({i \choose k}\beta^{i-k}   a_{i,j_2}X_1^{q{j_2}+pi-\ell}  \right) Y_1^{k}.
$$
\fi
since $q{j_1}+pi-\ell<q{j_2}+pi-\ell$, 
we know that for any fixed $k$, $a_{i,j_2}X^{j_2}Y^i$ always contributes 
strictly higher order of powers of $X_1$ than $a_{i,j_1}X^{j_1}Y^i$ in $f_1$.
\end{Remark}
\fi

\iflongversion
\begin{Remark}
\label{Remark:polygon}
Let $f(X, Y) := \sum_{i=0}^d \left(\sum_{j=0}^{\infty} a_{i, j}X^j \right) Y^i$.
For $0\leq i\leq d$, let $a_{i,j^*}$ be the first nonzero coefficient
among $\{a_{i,j}| 0\leq j < \infty\}$.
We observe that the Newton polygon of $f$ is completely 
determined by $a_{i,j^*}$, $0\leq i\leq d$.
\end{Remark}
\fi


\begin{Theorem}
\label{Theorem:finite-bound}
Let $f\in \C\langle X \rangle[Y]$.
Let $\tau\in{\N}_{>0}$.
Let $\sigma\in{\N}_{>0}$ and $g(T)=\sum_{k=0}^{\tau-1} b_k T^k$.
Assume that $(T^{\sigma}, g(T))$ is a Puiseux parametrization of $f$
of accuracy $\tau$.
Then one can compute a finite number $m\in \N$
such that $(T^{\sigma}, g(T))$ is a Puiseux parametrization of
accuracy $\tau$ of the approximation of $f$ of accuracy $m$.
We denote by ${\sf AccuracyEstimate}$ an algorithm to
compute such $m$ from $f$ and $\tau$.
\end{Theorem}
\begin{proof}
\iflongversion
Let $f_0 := f$, $X_0 := X$ and $Y_0 := Y$.
For $i=1, 2, \ldots$, 
Newton-Puiseux's algorithm 
computes numbers $q_{i}, p_{i}, {\ell}_{i}, \beta_{i}$
and the transformation
$$
f_i := X_i^{-\ell_{i}}f_{i-1}(X_i^{q_{i}}, X_i^{p_{i}}(\beta_{i} + Y_i))
$$
such that the assumption of Lemma~\ref{Lemma:one-iteration} is satisfied.

By Lemma~\ref{Lemma:one-iteration}, we know that for any $i$, 
a given number of terms of the coefficients of $f_i$ in $Y_i$
can be computed from a finite number of terms of the coefficients of $f_{i-1}$ in $Y_{i-1}$. 
Thus for any $i$, a given number of terms of the coefficients of $f_i$ in $Y_i$
can be computed from a finite number of terms of the coefficients of $f$ in $Y$.

On the other hand, the construction of Newton-Puiseux's algorithm
and Remark~\ref{Remark:polygon}
tell us that there exists a finite $M$, such that 
$\sigma$ and all the terms of $g(T)$
can be computed from a finite number 
of terms of the coefficients of $f_i$ in $Y_i$, $i=1,\ldots,M$.

Thus we conclude that there exists a finite number $m\in \N$
such that $(T^{\sigma}, g(T))$ is a Puiseux parametrization of
accuracy $\tau$ of the approximation of $f$ of accuracy $m$.
\fi

Next we show that there is an algorithm to compute $m$.
We initially set $m' := \tau$.
Let $f_0 := \sum_{i=0}^d \left(\sum_{j=0}^{m'} a_{i, j}X^j \right) Y^i$.
That is, $f_0$ is the approximation of $f$ of accuracy $m'+1$.
We run Newton-Puiseux's algorithm to check whether 
the terms $a_{k,m'}X^{m'}Y^k$, $0\leq k\leq d$,
make any contributions in constructing the Newton Polygons of all $f_i$.
If at least one of them make contributions,
we increase the value of $m'$  and restart the Newton-Puiseux's algorithm 
until
none of the terms $a_{k,m'}X^{m'}Y^k$, $0\leq k\leq d$,
makes any contributions in constructing Newton Polygons of all $f_i$.
\iflongversion
By Remark~\ref{Remark:monotone}, we can set $m := m'$.
\fi
\end{proof}

\begin{Lemma}
\label{Lemma:prime-chain}
Let $d, \tau\in{\N}_{>0}$.
Let $a_{i,j}$, $0\leq i\leq d$, $0\leq j < \tau $,
and $b_k$, $0\leq k < \tau$ be symbols.
Write 
$
{\bf a}=(a_{0,0},\ldots,a_{0,\tau-1},\\
\ldots,a_{d,0}, \ldots,a_{d,\tau-1})
$
and ${\bf b}=(b_0,\ldots,b_{\tau-1})$.
Let $f({\bf a}, X, Y)= \sum_{i=0}^d \left(\sum_{j=0}^{\tau-1} a_{i, j}X^j \right) Y^i\in\C[{\bf a}][X, Y]$.
Let $g({\bf b}, X)=
\sum_{k=0}^{\tau-1} \\
b_kX^k\in\C[{\bf b}][X]$.
Let $p := f({\bf a}, X, Y =g({\bf b}, X))$.
Let $F_k := \coeff{p, X^k}$, $0\leq k < \tau - 1\}$, and $F := \{F_0,\ldots, F_{\tau-1}\}$.
Then under the order ${\bf a} < {\bf b}$ and $b_0<b_1<\cdots <b_{\tau-1}$, 
$F$ forms a zero-dimensional regular chain in $\C({\bf a})[{\bf b}]$
with main variables $(b_0,b_1,\ldots,b_{\tau-1})$
and main degrees $(d, 1,\ldots,1)$.
In addition, we have
\begin{itemizeshort}
\item $F_0=\sum_{i=0}^d a_{i,0}b_0^i$ and 
\item $\init{F_1}=\cdots=\init{F_{\tau-1}}=\sum_{i=1}^d i\cdot a_{i, 0}b_0^{i-1}$.
\end{itemizeshort}
\end{Lemma}
\begin{proof}
Write $p=\sum_{i=0}^d \left(\sum_{j=0}^{\tau-1} a_{i, j}X^j\right) \left(  \sum_{k=0}^{\tau-1} b_kX^k \right)^i $ 
as a univariate polynomial in $X$.
Observe that $F_0=\sum_{i=0}^d a_{i,0}b_0^i$.
Therefore $F_0$ is irreducible in $\C({\bf a})[{\bf b}]$.
Moreover, we have $\mvar{F_0}=b_0$ and $\mdeg{F_0}=d$.

Since $d>0$, we know that $a_{1, 0} \left(  \sum_{k=0}^{\tau-1} b_kX^k \right)$
appears in $p$.
Thus, for $0\leq k <\tau$, 
$b_k$ appears in $F_k$.
Moreover, for any $k\geq 1$ and $i<k$, 
$b_k$ can not appear in $F_i$ since $b_k$ and $X^k$
are always raised to the same power.
For the same reason, for any $i>1$, $b_k^i$
cannot appear in $F_k$, for $1\leq k <\tau$.
Thus $\{F_0,\ldots,F_{\tau-1}\}$ is a triangular set with
main variables $(b_0,b_1,\ldots,b_{\tau-1})$
and main degrees $(d, 1,\ldots,1)$.

Moreover, we have $\init{F_1}=\cdots=\init{F_{\tau-1}}=\sum_{i=1}^d i\cdot a_{i, 0}b_0^{i-1}$, 
which is coprime with $F_0$.
Thus $F= \{F_0,\ldots, F_{\tau-1}\}$ is a regular chain. 
\end{proof}

\begin{Lemma}
\label{Lemma:generic}
Let $f= \sum_{i=0}^d \left(\sum_{j=0}^{\infty} a_{i, j}X^j \right) Y^i\in\C[[X]][Y]$. 
Assume that $\deg(f, Y)>0$ and $f$ is general in $Y$.
Let $\varphi(X)=\sum_{k=0}^{\infty} b_kX^k\in\C[[X]]$ 
such that $f(X, \varphi(X))=0$ holds.
Let $\tau>0\in\N$. 
Then ``generically'', $b_i$, $0\leq i < \tau$, 
 can be completely determined by $\{a_{i,j}\mid 0\leq i\leq d, 0\leq j<\tau\}$.
\end{Lemma}
\begin{proof}
\iflongversion
By $f(X, Y)=0$, we know that $f(X, Y)=0~\mbox{mod}~\langle X^{\tau}\rangle$.
Therefore, we have
$$
\sum_{i=0}^d \left(\sum_{j<\tau} a_{i, j} X^{j} \right)\left(\sum_{k<\tau} b_kX^k \right)^i=0 ~\mbox{mod}~\langle X^{\tau}\rangle.
$$
Let $p=\sum_{i=0}^d \left(\sum_{j<\tau} a_{i, j} X^{j} \right)\left(\sum_{k<\tau} b_kX^k \right)^i$.
Let $F_i := \{\coeff{p, X^i}$, $0\leq i < \tau\}$, and $F := \{F_0,\ldots, F_{\tau-1}\}$.
Since $f$ is general in $Y$ and $f(X, \varphi(X))=0$, there exists $i^*>0$
such that $a_{i^*,0}\neq 0$. 
By Lemma~\ref{Lemma:prime-chain}, we have $F_0=\sum_{i=0}^d a_{i,0}b_0^i$.
Thus $b_0$ can be completely determined by $a_{i,0}$, $0\leq i\leq d$.
In order to completely determine $b_1,\ldots,b_{\tau-1}$, 
it is enough to gurantee $\res{F_0, F_i, b_0}\neq 0$ holds. 
Therefore the values of $b_k$, $0\leq k < \tau$ can be completely determined
from almost all the values 
of $a_{i,j}$, $0\leq i\leq d$, $0\leq j < \tau $.
\fi
\end{proof}


\section{Accuracy estimates}

\label{sec:bounds}

Let $R := \{r_1(X_1,X_2),\ldots, 
r_{s-1}(X_1,\ldots,X_s)\}\subset\C[X_1<\cdots<X_s]$
be a strongly normalized regular chain.
In this section, we show that to compute 
the limit points of $W(R)$, 
it suffices to compute the Puiseux parametrizations
of $R$ of some accuracy.
Moreover, we provide accuracy estimates in Theorem~\ref{Theorem:bound-general}.

\begin{Lemma}
  \label{Lemma:general}
  Let $f=a_d(X)Y^d+\cdots+a_0(X)\in\C\langle X\rangle[Y]$,
  where $d>0$ and $a_d(X)\neq 0$.  For $0\leq i\leq d$, let $\delta_i
  := \ord{a_i}$.  Let $k := \min{\delta_0,\ldots,\delta_d}$.  Let
  $\widetilde{f} := f/X^k$.  Then we have $\widetilde{f}\in\C\langle
  X\rangle[Y]$ and $\widetilde{f}$ is general in $Y$.
  This process of producing $\widetilde{f}$ from $f$ is called {``\em
    making $f$ general''} and denote by {\sf MakeGeneral} an operation
  which produces $\widetilde{f}$ from $f$.
\end{Lemma}

\begin{proof}
\iflongversion
  Since $k= \min{\delta_0,\ldots,\delta_d}$, there exists $i$, $1\leq
  i\leq d$, such that $k=\delta_i$.  Moreover, for all $1\leq j\leq
  d$, we have $\delta_j \geq k$.  Thus for every such $i$, we have
  $\ord{a_i(X)/X^k}=0$ and $a_j(X)/X^k\in \C\langle X\rangle$, $0\leq
  j\leq d$.  This shows that $\widetilde{f}\in\C\langle X\rangle[Y]$
  and $\widetilde{f}$ is general in $Y$.
\fi
\end{proof}

The following lemma shows that
computing limit points reduces 
to making a polynomial $f$ general.

\begin{Lemma}
\label{Lemma:Walker}
Let $f\in\C\langle X\rangle[Y]$, where $\deg(f, Y)>0$.
Assume that $f$ is general in $Y$.
Let $\rho>0$ be small enough such that
$f$ converges in $\abs{X}<\rho$.
Let $V_{\rho}(f) := \{(x,y)\in\C^2\mid 0<\abs{x}<\rho, f(x,y)=0\}$.
Then we have ${\rm lim}_0(V_{\rho}(f))=\{(0, y)\in\C^2\mid f(0,y)=0\}$.
\end{Lemma}
\begin{proof}
Let $(X=T^{\varsigma_i}, Y=\varphi_i(T))$, $1\leq i\leq c\leq d$, 
be the Puiseux parametrizations of $f$.
By Lemma~\ref{Lemma:V0} and Theorem~\ref{thrm:PuiseuxParamRC}, 
we have ${\rm lim}_0(V_{\rho}(f))=\cup_{i=1}^c\{(0, y)\in\C^2\mid y=\varphi_i(0)\}$.
Let $(X=T^{\sigma_i}, g_i(T))$, $i=1,\ldots,c$, be the corresponding 
Puiseux parametrizations of $f$ of accuracy $1$.
By Theorem~\ref{Theorem:finite-bound},
there exists an approximation $\widetilde{f}$ of $f$
of some finite accuracy such that $(X=T^{\sigma_i}, g_i(T))$, $i=1,\ldots,c$,
are also Puiseux parametrizations of $\widetilde{f}$ of accuracy $1$.
Thus, we have $\varphi_i(0)=g_i(0)$, $i=1,\ldots,c$.
Since $\widetilde{f}$ is also general in $Y$, by Theorem 2.3 of Walker~\cite{RW78}, we have
$\cup_{i=1}^c\{(0, y)\in\C^2\mid y=g_i(0)\}=\{(0, y)\in\C^2\mid \widetilde{f}(0,y)=0\}$.
Since $\widetilde{f}(0,y)=f(0,y)$, the Lemma holds.
\end{proof}

\begin{Lemma}
  \label{Lemma:finite}
  Let $a(X_1,\ldots,X_s)\in\C[X_1,\ldots,X_s]$.  Let
  $g_i=\sum_{j=0}^{\infty}c_{ij}T^j\in \C\langle T \rangle$.  
  We write
  $a(g_1,\ldots,g_s)$ as $\sum_{k=0}^{\infty}b_kT^k$.  To compute
  a given $b_k$, one only needs the set of coefficients $\{c_{i, j}\mid 1\leq
  i\leq s, 0\leq j\leq k\}$.
\end{Lemma}
\begin{proof}
\iflongversion
  We observe that any $c_{i, j}$, where $j > k$, does not make any
  contribution to $b_k$.
\fi
\end{proof}


\begin{Lemma}
\label{Lemma:general-generic}
 Let $f=a_d(X)Y^d+\cdots+a_0(X)\in\C\langle X\rangle[Y]$, where $d>0$, 
and $a_d(X)\neq 0$.
Let $\delta := \ord{a_d(X)}$.
Then ``generically'', a Puiseux parametrization of $f$ of accuracy $\tau$
can be computed from an approximation of $f$ of accuracy $\tau+\delta$.
\end{Lemma}
\begin{proof}
Let $\widetilde{f} := {\sf MakeGeneral}(f)$. 
Observe that $f$ and $\widetilde{f}$ have the same system of 
Puiseux parametrizations. 
Then the conclusion follows from Lemma~\ref{Lemma:general} and Lemma~\ref{Lemma:generic}.
\end{proof}

\begin{Theorem}
\label{Theorem:bound-general}
Let $R := \{r_1(X_1,X_2),\ldots, 
r_{s-1}(X_1,\ldots,X_s)\}\subset\C[X_1<\cdots<X_s]$
be a strongly normalized regular chain.
For $1\leq i\leq s-1$, let $h_i := \init{r_i}$, $d_i := \deg(r_i, X_{i+1})$
and $\delta_i := \ord{h_i}$.
We define $f_i$, $2\leq i\leq s-1$, 
and $\varsigma_j$, $T_j$, $\varphi_j(T_j)$, $1\leq j\leq s-2$, as follows
\begin{itemizeshort}
\item Let $(X_1=T_1^{\varsigma_1}, X_2=\varphi_1(T_1))$ be a Puiseux parametrization of $r_1(X_1,X_2)$.
\item Let $f_i := r_i(X_1= T_1^{\varsigma_1},X_2=\varphi_1(T_1), \ldots,X_i=\varphi_{i-1}(T_{i-1}), \\
X_{i+1})$.
\item Let $(T_{i-1}=T_i^{\varsigma_i}, X_{i+1}=\varphi_i(T_i))$ be a Puiseux parametrization of $f_i$.
\end{itemizeshort}

Then we have the following results:

\begin{itemizeshort}
\item[$(i)$] Let $T_0 := X_1$, for $0\leq i\leq s-2$, 
define $g_i(T_{s-2}):=T_{s-2}^{\prod_{k=i+1}^{s-2}\varsigma_k}$, 
then we have $T_i =  g_i(T_{s-2})$.
\item[$(ii)$] We have $f_{s-1}\in\C\langle T_{s-2}\rangle[X_{s}]$.

\item[$(iii)$] There exist numbers $\tau_1,\ldots,\tau_{s-2}\in\N$
such that in order to make $f_{s-1}$ general in $X_{s}$, 
it suffices to compute the polynomial parts of $\varphi_i$ 
of accuracy $\tau_i$, $1\leq i\leq s-2$.
Moreover, if we write the algorithm ${\sf AccuracyEstimate}$ for short as $\theta$,
the accuracies $\tau_i$ can be computed in the following manner
\begin{itemizeshort}
\item let $\tau_{s-2} := (\prod_{k=1}^{s-2} \varsigma_k)\delta_{s-1}+1$
\item let $\tau_{i-1} := \max{\theta(f_i, \tau_i), (\prod_{k=1}^{i-1} \varsigma_k)\delta_{s-1}+1}$, 
for $s-2 \geq i \geq 2$.
\end{itemizeshort}

\item[$(iv)$] 
Generically, for $1\leq i\leq s-3$, we can choose $\tau_i=(\prod_{k=1}^{s-2} \varsigma_k)(\sum_{k=2}^{s-1}\delta_i)+1$.
\item[$(v)$] The indices $\varsigma_k$ can be replaced with $d_k$, $k=1,\ldots,s-2$.
\end{itemizeshort}
\end{Theorem}

\begin{proof}
We prove $(i)$ by induction.
Clearly $(i)$ holds for $i=s-2$.
Suppose it holds for $i$.
Then we have
$$
T_{i-1}  =  T_i^{\varsigma_i}
        =  \left(      T_{s-2}^{\prod_{k=i+1}^{s-2}\varsigma_k}     \right)^{\varsigma_i}
        =  \left(     T_{s-2}^{\prod_{k=i}^{s-2}\varsigma_k}     \right)
$$
Therefore $(i)$ holds also for $i-1$.
So $(i)$ holds for all $0\leq i\leq s-2$.

\iflongversion
Note that 
\begin{equation}
\label{equation:fs-1}
\begin{array}{rcl}
f_{s-1} &=& r_{s-1}(X_1= T_1^{\varsigma_1},X_2=\varphi_1(T_1), \ldots, \\
       &&  ~\qquad X_{s-2}=\varphi_{s-3} (T_{s-3}), X_{s-1}=\varphi_{s-2}(T_{s-2}), X_{s})\\
       &=& r_{s-1}(X_1=g_0(T_{s-2}),X_2=\varphi_1(g_1(T_{s-2})), \ldots,  \\
       && ~\qquad X_{s-2}=\varphi_{s-3}(g_{s-3}(T_{s-2})), \\
       && ~\qquad X_{s-1}=\varphi_{s-2}(g_{s-2}(T_{s-2})), X_{s})\\
\end{array}
\end{equation}
Since $\ord{g_i(T_{s-2})}>0$ for all $0\leq i\leq s-2$, 
by Lemma~\ref{Lemma:GFsubs}, $(ii)$ holds.
\fi

Note that $g_0(T_{s-2})= T_{s-2}^{\prod_{k=1}^{s-2}\varsigma_k}$.
Since $\ord{h_{s-1}(X_1)}=\delta_{s-1}$, 
we have 
$$
\ord{h_{s-1}(X_1=g_0(T_{s-2}))}=\left(\prod_{k=1}^{s-2} \varsigma_k\right)\delta_{s-1}.
$$
Let $\tau_{s-2} := (\prod_{k=1}^{s-2} \varsigma_k)\delta_{s-1}+1$.
By Lemma~\ref{Lemma:general}, to make $f_{s-1}$ general in $X_s$, 
it suffices to 
compute the polynomial parts of the coefficients of
$f_{s-1}$ of accuracy $\tau_{s-2}$.

By Lemma~\ref{Lemma:finite},
\iflongversion
and Equation~(\ref{equation:fs-1}), 
\fi
we need to compute the polynomial parts of $\varphi_i(g_i(T_{s-2}))$, $1\leq i\leq s-2$, 
of accuracy $\tau_{s-2}$.
Since $\ord{g_i(T_{s-2})}=\prod_{k=i+1}^{s-2}\varsigma_k$, 
to achieve this accuracy, 
it's enough to compute the polynomial parts of $\varphi_i$
of accuracy $(\prod_{k=1}^i \varsigma_k)\delta_{s-1} + 1$, for $1\leq i\leq s-2$.

On the other hand, 
since $f_i = r_i(X_1= T_1^{\varsigma_1},X_2=\varphi_1(T_1), 
\ldots, X_i=\varphi_{i-1}(T_{i-1}), X_{i+1})$
and $(T_{i-1}=T_i^{\varsigma_i}, X_{i+1}=\varphi_i(T_i))$ is a Puiseux parametrization of $f_i$,
by Theorem~\ref{Theorem:finite-bound} and Lemma~\ref{Lemma:finite},
to compute the polynomial part of $\varphi_i$ of accuracy $\tau_i$, 
we need the polynomial part of $\varphi_{i-1}$ of accuracy $\theta(f_i,\tau_i)$.

Thus, take $\tau_{s-2} := (\prod_{k=1}^{s-2} \varsigma_k)\delta_{s-1}+1$ and 
$\varphi_{i-1}=
\max{\theta(f_i, \tau_i),(\prod_{k=1}^{i-1} \varsigma_k)\delta_{s-1}+1}$
for $2\leq i\leq s-2$ will guarantee $f_{s-1}$ can be made general in $X_{s}$. 
So $(iii)$ holds. 
%
By Lemma~\ref{Lemma:general-generic}, 
generically we can choose $\theta(f_i, \tau_i)=\tau_i+(\prod_{k=1}^{i-1}\sigma_k)\delta_i$, $2\leq i\leq s-2$.
Therefore $(iv)$ holds.
Since we have $\varsigma_k\leq d_k$, $1\leq k\leq s-2$, $(iv)$ holds.
\end{proof}

 \section{Algorithm}
\label{sec:principle}

In this section, we provide a complete algorithm 
for computing the non-trivial limit points 
of the quasi-component of a one-dimensional strongly normalized regular chain
based on the results of the previous sections.

\begin{algorithm}
\label{algorithm:limitatzero}
\linesnumbered
\caption{${\sf LimitPointsAtZero}$}
\KwIn{$R := \{r_1(X_1,X_2),\ldots, r_{s-1}(X_1,\ldots,X_s)\}\subset\C[X_1<\cdots<X_s]$, $s>1$,
is a strongly normalized regular chain.
}
\KwOut{
The non-trivial limit points of $W(R)$ whose $X_1$-coordinates are $0$.
}
\Begin{
let $S := \{(T_0)\}$\; 
compute the accuracy estimates $\tau_1,\ldots,\tau_{s-2}$ by Theorem~\ref{Theorem:bound-general}; 
let $\tau_{s-1}=1$\;
\For{$i$ from $1$ to $s-1$}{
    $S' := \emptyset$\;
    \For{$\Phi\in S$}{
      $f_i := r_i(X_1=\Phi_1,\ldots,X_{i}=\Phi_i, X_{i+1})$\;
      \If{$i>1$}{
         let $\delta := \ord{f_i, T_{i-1}}$;
         let $f_i := f_i/T_{i-1}^{\delta}$\;
      }
      $E := \NewtonPuiseux{f_i, \tau_i}$\;
      \For{$(T_{i-1}=\phi(T_i), X_{i+1}=\varphi(T_i))\in E$}{
           $S' :=  S'\cup \{\Phi(T_{i-1}=\phi(T_i))\cup (\varphi(T_i))\}$
      }
    }
    $S := S'$
}
   \lIf{$S=\emptyset$}{
        return $\emptyset$
   }
   \Else{
        return $\eval{S, T_{s-1}=0}$
   }
}
\end{algorithm}

\begin{algorithm}
\label{algorithm:limit}
\linesnumbered
\caption{${\sf LimitPoints}$}
\KwIn{
A strongly normalized regular chain\\
$R := \{r_1(X_1,X_2),\ldots, r_{s-1}(X_1,\ldots,X_s)\}\subset\C[X_1<\cdots<X_s]$, $s>1$.
}
\KwOut{
All the non-trivial limit points of $W(R)$.
}
\Begin{
    let $h_R := \init{R}$;
    let $L$ be the set of roots of $h_R$\;
    $S := \emptyset$\;
    \For{$\alpha\in L$}{
      $R_{\alpha} := R(X_1=X_1+\alpha)$\;
      $S_{\alpha} := \LimitPointsAtZero{R_{\alpha}}$\;
      update $S_{\alpha}$ by replacing the first coordinate of every point in it by $\alpha$\;
      $S := S\cup S_{\alpha}$
    }
    return $S$
}
\end{algorithm}

\begin{Remark}
Note that line $9$ of Algorithm~\ref{algorithm:limitatzero}
computes Puiseux parametrizations of $f_i$ of accuracy $\tau_i$.
Thus $(\phi(T_i),\varphi(T_i))$ at line $10$
cannot have negtive orders.

If the D5 principle is applied to Algorithms~\ref{algorithm:limitatzero}
and~\ref{algorithm:limit}, the limit points of $W(R)$
can be represented by a finite family of regular chains.
\end{Remark}

\begin{Proposition}
\label{Proposition:newton}
Algorithm~\ref{algorithm:limit} is correct and terminates.
\end{Proposition}
\begin{proof}
It follows from Theorem~\ref{thrm:PuiseuxParamRC}, Theorem~\ref{Theorem:finite-bound}, Theorem~\ref{Theorem:bound-general}
and Lemma~\ref{Lemma:Walker}.
\end{proof}

\section{Experimentation}
\label{section}

We have implemented Algorithm~\ref{algorithm:limit} of Section~\ref{sec:principle},
which computes the limit points of the quasi-component of 
a one-dimensional strongly normalized regular chain.
The implementation is based on the library {\RegularChains}
and the command {\sf algcurves[puiseux]} of {\Maple}.
The code is available at \url{http://www.orcca.on.ca/~cchen/ACM13/LimitPoints.mpl}.
This preliminary implementation relies on algebraic factorization,
whereas, as suggested in~\cite{Duv89}, 
applying the D5 principle, in the spirit of triangular
decomposition algorithms, for instance~\cite{CM12}, would be sufficient
when computations need to split into different cases.
This would certainly improve performance greatly and this
enhancement is work in progress.

As pointed out in the introduction, 
the computation of the limit points of 
the quasi-component of a regular chain 
can be applied to removing redundant
components in a Kalkbrener triangular decomposition.
In Table~\ref{table:redundant}, we report on experimental results 
of this application. 

The polynomial systems listed in this table are one-dimensional 
polynomial systems selected from the literature~\cite{CGLMP07,CM12}.
For each system, we first call the {\sf Triangularize} command 
of the library {\RegularChains}, 
with the option ``{\sf 'normalized='strongly', 'radical'='yes'}''.
For the input system,
this process computes a Kalkbrener triangular decomposition ${\cal R}$
where the regular chains are strongly normalized and
their saturated ideals are radical.
Next, for each one-dimensional regular chain $R$ in the output, 
we compute the limit points $\lim(W(R))$,
thus deducing a set of regular chains $R_1, \ldots, R_e$
such the union of their quasi-components equals 
the Zariski closure $\overline{W(R)}$.
The algorithm {\sf Difference}~\cite{CGLMP07} is then called 
to test whether or not there exists a pair $R, R'$ of regular chains
of ${\cal R}$ such that the inclusion 
$\overline{W(R)} \, \subseteq \, \overline{W(R')}$ holds.

In Table~\ref{table:redundant}, the column T and \#(T) denote respectively 
the timings spent by {\sf Triangularize} and 
the number of regular chains returned by this command;
the column d-1 and d-0 denote respectively the number of 
$1$-dimensional and $0$-dimensional regular chains, 
whose sum is exactly \#(T); 
the column R and \#(R) denote respectively the timings 
spent on removing redundant components in the output of 
{\sf Triangularize} and the number of regular chains
in the output irredundant decomposition.
As we can see in the table, most of the decompositions are 
checked to be irredundant, which we could not do before
this work by means of triangular decomposition algorithms.
In addition, the three redundant 
$0$-dimensional components in the Kalkbrener triangular decomposition 
of system f-744 are successfully removed.
Therefore, we have verified experimentally the benefits provided
by the algorithms presented in this paper.

\begin{table}
\centering
\caption{Removing redundant components.}
\label{table:redundant}
\begin{tabular}{|c|c|c|c|c|c|c|}\hline
Sys   & T & \#(T)  &d-1  & d-0 & R  & \#(R) \\\hline
f-744 &           14.360 & 4 & 1 & 3 & 432.567 & 1  \\
Liu-Lorenz &      0.412 & 3 & 3 & 0 & 216.125 & 3  \\
MontesS3 &        0.072 & 2 & 2 & 0 & 0.064 & 2  \\
Neural &          0.296 & 5 & 5 & 0 & 1.660 & 5  \\
Solotareff-4a &           0.632 & 7 & 7 & 0 & 32.362 & 7  \\
Vermeer &         1.172 & 2 & 2 & 0 & 75.332 & 2  \\
Wang-1991c &      3.084 & 13 & 13 & 0 & 6.280 & 13  \\
\hline
\end{tabular}
\end{table}

\section{Concluding remarks}
\label{sec:discussion}

We conclude with a few remarks
about special cases and a generalization of the
algorithms presented in this paper.

\smallskip\noindent{\small \bf Reduction to  strongly normalized chains.}
Using the hypotheses of Lemma~\ref{Lemma:limitWR3},
we observe that one can reduce the computation
of \limit{W(R)} to that of \limit{W(N)}.
Indeed, under the assumption that \sat{R} has dimension one,
both \limit{W(R)} and \limit{W(N)} are finite.
Once the set \limit{W(N)} is computed, one can easily check which points
in \limit{W(N)} do not belong to $W(R)$ and then deduce \limit{W(R)}.
This reduction to strongly normalized regular chains 
has the advantage that $h_N$ is a univariate polynomial
in ${\C}[X_1]$, which simplifies the presentation of the basic ideas 
of our algorithms, see Section~\ref{sec:limitpoints}.
However, it has two drawbacks. First the coefficients of $N$
are generally much larger than those of $R$. Secondly,
\limit{W(N)} may also be much larger than \limit{W(R)}.
A detailed presentation of a direct computation of \limit{W(R)},
without reducing to \limit{W(N)}, will be done in a future paper.

\smallskip\noindent{\small \bf Shape lemma case.}
Here, by reference to the paper~\cite{BeckerMoraMarinariTraverso1994}
\iflongversion
(which deals with polynomial ideals of dimension zero)
\fi
we assume that, for $2 \leq i \leq e$, the polynomial $r_i$
involves only the variables $X_1, X_2, X_i$ and that ${\deg}(r_i, X_i) = 1$ holds.
In this case, computing Puiseux series expansions is required
only for the polynomial of $R$ of lower rank, namely $r_1$.
In this case, the algorithms presented in this paper are 
much simplified.
However, for the specific purpose of 
solving polynomial systems via triangular decompositions,
reducing to this Shape lemma case, via a random change of
coordinates, has a negative impact on performance
and software design, for many problems of practical interest.
In contrast, the point of view of the work initiated in this paper 
is two-fold: first, deliver algorithms that do not require
any genericity assumptions; second develop criteria that 
take advantage of specific properties of the input systems
in order to speedup computations.
Yet, in our implementation, several tricks are used to avoid
unnecessary Puiseux series expansions, such as 
applying the theorem (see \cite{GF01} p.113) 
on the continuity of the roots of a parametric polynomial.

\smallskip\noindent{\small \bf  Handling the case 
where \sat{R} has dimension greater than 1.}
From now on, \sat{R} has dimension $s-e \geq 2$.
We use the notations of Lemma~\ref{Lemma:limitWR5}
and recall that each point of \limit{W(R' \cup r_e)}
is in particular a point of \limit{W(R')}.
Since we know how to compute \limit{W(R')}
when $R'$ consists of a single polynomial,
we assume, by induction that a triangular decomposition
of  \limit{W(R')} has been computed
in the form ${W(R_1)} \, \cup \, \cdots \, \cup \,  {W(R_f)}$
for regular chains $R_1, \ldots, R_f$.

We observe that a point $p \in \limit{W(R')}$ can be ``extended''
to a point of \limit{W(R' \cup r_e)} in two ways.
First, if $p$ does not cancel the initial of $r_e$
(which can be tested algorithmically),
then, by applying the theorem (see again \cite{GF01} p.113) 
on the continuity of the roots
to $r_e$, we extend $p$ with the $X_s$-roots of $r_e$,
after specializing $(X_1,  \ldots, X_{s-1})$ to $p$.
From now on, we assume that $p$ cancels the initial $h_e$ of $r_e$.
In this case, we compute a truncated Puiseux parametrization
about $p$ using the regular chain $R_i$ such that 
$p \in W(R_i)$ holds.
After substitution into the polynomial $r_e$,
we apply Puiseux theorem and compute the 
limit points of $W(R' \cup r_e)$ extending $p$, in the manner
of the algorithms of Section~\ref{sec:principle}.

There are new challenges, however, w.r.t. to the one-dimensional case.
First, parametrizations may involve now more than one parameter.
When this happens, one should use  
Jung-Abhyankar theorem~\cite{APGR00}
instead of the Puiseux theorem.
The second difficulty is that $\limit{W(R')} \, \cap \, V(h_e)$
may be infinite.
This will not happen, however, if \sat{R'} has dimension at most $2$
and $h_e$ is regular w.r.t. \sat{R'}.
This second assumption can be regarded as a genericity assumption.
Thus the algorithms presented can easily be extended to dimension two,
under that assumption, which can be tested algorithmically.
Overcoming in higher dimensions this cardinality issue 
with $\limit{W(R')} \, \cap \, V(h_e)$, 
requires to understand which ``configurations'' are essentially the same.
Since  \limit{W(R)}, as an algebraic set, 
can be described by finitely many regular chains,
this is, indeed, possible and work in progress.



\bibliographystyle{plain}

\end{document}